\documentclass[times]{acsauth}

\usepackage{moreverb}

\usepackage[colorlinks,bookmarksopen,bookmarksnumbered,citecolor=red,urlcolor=red]{hyperref}

\newcommand\BibTeX{{\rmfamily B\kern-.05em \textsc{i\kern-.025em b}\kern-.08em
T\kern-.1667em\lower.7ex\hbox{E}\kern-.125emX}}

\def\volumeyear{2010}

\usepackage{subfig}
\usepackage{graphicx}

\newtheorem{remark}{Remark}
\newtheorem{lemma}{Lemma}

\newtheorem{proposition}{Proposition}

\def\begequarr{\begin{eqnarray}}
\def\endequarr{\end{eqnarray}}
\def\begequarrs{\begin{eqnarray*}}
\def\endequarrs{\end{eqnarray*}}
\def\begarr{\begin{array}}
\def\endarr{\end{array}}
\def\begequ{\begin{equation}}
\def\endequ{\end{equation}}
\def\lab{\label}
\def\begdes{\begin{description}}
\def\enddes{\end{description}}
\def\begenu{\begin{enumerate}}
\def\begite{\begin{itemize}}
\def\endite{\end{itemize}}
\def\endenu{\end{enumerate}}

\def\lef[{\left[\begin{array}}
\def\rig]{\end{array}\right]}
\def\qed{\hfill$\Box \Box \Box$}
\def\begcen{\begin{center}}
\def\endcen{\end{center}}
\def\begrem{\begin{remark}\rm}
\def\endrem{\end{remark}}

\def\begmat#1{\begin{bmatrix}#1\end{bmatrix}}
\def\begali#1{\begin{align}{#1}\end{align}}
\def\begalis#1{\begin{align*}{#1}\end{align*}}

\def\cale{{\cal E}}
\def\cali{{\cal I}}

\def\calb{{\cal B}}

\def\cale{{\cal E}}

\def\calp{{\cal P}}
\def\call{{\cal L}}

\def\L2e{{\cal L}_{2e}}

\def\rea{\mathbb{R}}

\def\col{\mbox{col}}
\def\hal{{1 \over 2}}

\def\vida{u_{\tt{ES}}}
\def\vpd{u_{\tt{EPD}}}

\def\L2e{{\cal L}_{2e}}

\def\rea{\mathbb{R}}

\def\col{\mbox{col}}
\def\hal{{1 \over 2}}

\def\cl{{\rm cl}}

\def\IJRNLC{{\it Int. J. of Robust and Nonlinear Control}}
\def\TAC{{\it IEEE Trans. Automatic Control}}

\def\EJC{{\it European Journal of Control}}

\def\CDC{{\it IEEE Conference on Decision and Control}}
\def\SCL{{\it Systems \& Control Letters}}
\def\AUT{{\it Automatica}}

\usepackage{color}

\begin{document}

\runningheads{B. Yi~\MakeLowercase{\it et al}.}{Smooth, time-invariant regulation of nonholonomic systems}

\title{Smooth, time-invariant regulation of nonholonomic systems via energy pumping-and-damping\footnotemark[2]}

\author{Bowen Yi$^{1,2}$, Romeo Ortega$^{2,3}$, Weidong Zhang$^1$\corrauth}

\address{\center
{\rm 1.} Department of Automation, Shanghai Jiao Tong University, Shanghai 200240, China
\\
{\rm 2.} Laboratoire des Signaux et Syst\'emes, CNRS-CentraleSup\'elec, Gif-sur-Yvette 91192, France
\\
{\rm 3.} Department of Control Systems and Informatics, ITMO University, St. Petersburg 197101, Russia
}

\corraddr{Shanghai Jiao Tong University, Shanghai 200240, China (\texttt{wdzhang@sjtu.edu.cn})}

\begin{abstract}
In this paper we propose an energy pumping-and-damping technique to regulate nonholonomic systems described by kinematic models. The controller design follows the widely popular interconnection and damping assignment passivity-based methodology, with the free matrices partially structured. Two asymptotic regulation objectives are considered: drive to {\it zero} the state or drive the systems {\it total energy} to a desired constant value. In both cases, the control laws are {\it smooth, time-invariant}, state-feedbacks.  For the nonholonomic integrator we give an almost global solution for both problems, with the objectives ensured for all system initial conditions starting outside a set that has zero Lebesgue measure and is nowhere dense. For the general case of higher-order nonholonomic systems in chained form, a local stability result is given. Simulation results comparing the performance of the proposed controller with other existing designs are also provided.
\end{abstract}


\footnotetext[2]{the National Natural Science Foundation of China (61473183, U1509211, 61627810), National Key R\&D Program of China (SQ2017YFGH001005), China Scholarship Council, and by the Government of the Russian Federation (074U01), the Ministry of Education and Science of Russian Federation (14.Z50.31.0031, goszadanie no. 8.8885.2017/8.9).}

\maketitle

\section{Introduction}
\lab{sec1}
The study of mechanical system subject to nonholonomic constraints has been carried-out within the realm of  analytical mechanics \cite{Bloch2003,Borisov2002}. The complexity and highly nonlinear dynamics of nonholonomic mechanical systems  make the motion control problem challenging \cite{Bloch2003}. A key feature that distinguishes the control of nonholonomic systems from that of holonomic systems is that in the former, it is not possible to render \emph{asymptotically stable} an isolated equilibrium with a smooth (or even continuous), time-invariant (static or dynamic), state-feedback control law. The best one can achieve with smooth control laws is to stabilise an equilibrium manifold \cite{MASSCH}. This obstacle stems from Brockett's necessary condition for asymptotic stabilization \cite{BRO}---see also \cite{Bloch2003}. In view of the aforementioned limitation, time-varying feedback \cite{POM,JIAetal}, discontinuous feedback \cite{AST,FUJetal}, switching control methods \cite{LIB} and hybrid systems approaches \cite{HESMOR}, have been considered in the control literature. In this paper we are interested in investigating the possibilities of regulating nonholonomic systems via smooth, time-invariant state-feedback.

In the 1998 paper \cite{ESCetal} a radically new approach to regulate the behaviour of nonholonomic systems was proposed. The work was inspired by the classical field-oriented control (FOC) of induction motors, which was introduced in the drives community in 1972 \cite{BLA}, and is now the \emph{de facto} standard in all high-performance applications of electric drives---see \cite{MARbook} for a modern control-oriented explanation of the method. The basic idea of FOC is to regulate, with a smooth, time invariant, state-feedback law, the speed (or the torque) of the motor by inducing an oscillation, with the desired \emph{frequency and amplitude}, to the motors magnetic flux, that is a two-dimensional vector. From the physical viewpoint this is tantamount to controlling the mechanical energy via the regulation of the magnetic energy. As shown in \cite{ESCetal}, applying this procedure to the nonholonomic integrator allows us to drive the state to an arbitrarily small neighborhood of the origin as well as solving trajectory tracking problems. Unfortunately, when the objective is to drive the state to \emph{zero}, the control law includes the division by a state-dependent signal---rendering the controller not-globally defined. Although this signal is bounded away from zero \emph{along trajectories}, in the face of noise or parameter uncertainty, it may cross through zero, putting a question mark on the robustness of the design. It should be mentioned that the results of \cite{ESCetal} were later adopted in \cite{DIXetal} and are the inspiration for the transverse function approach pursued in \cite{MORSAMtac,MORSAM}.

In \cite{YIetal} it is shown that FOC can be interpreted as an Interconnection and Damping Assignment Passivity-based controller (IDA-PBC) \cite{ORTetal} that assigns a port-Hamiltonian (pH) structure to the closed-loop. The corresponding energy function has the shape of a ``Mexican sombrero", whose minimum is achieved in the periodic orbit that we want to reach, \emph{e.g.}, $H_\ell(x_\ell)=\beta_\ell$, with $x_\ell$ part of the state coordinates, whose energy function is $H_\ell(x_\ell)$, and $\beta_\ell$ is a positive, tuning constant---see Figure \ref{figsom}. The same approach was proposed in \cite{GOMetal} to induce an oscillation in the Ball-and-Beam system and in \cite{DUISTR} in walking robot applications. To assign the Mexican sombrero shape the energy function contains a term of the form $(H_\ell(x_\ell)-\beta_\ell)^2$, whose gradient can be transferred to the dissipation matrix of the pH system, giving then an interpretation of ``Energy Pumping-and-Damping" (EPD). That is, a controller that injects or extracts energy from the system depending on the location of the state with respect to the desired oscillating trajectory, see Fig \ref{fig_pd}.  This point of view was adopted in \cite{ASTetal} to design a controller that swings up---without switching---the cart-pendulum system. In the sequel, we will refer to this controller design technique as \emph{EPD IDA-PBC}, that is, a variation of IDA-PBC where the (otherwise free) dissipation matrix is partially structured. EPD IDA-PBC  has been used in \cite{YIetal} to solve the more general \emph{orbital stabilization} problem, where we made the important observation that, by setting $\beta_\ell=0$, we can achieve \emph{regulation to zero} of the state.\\

The main objective of this paper is to show that an EPD IDA-PBC formulation of the scheme proposed in \cite{ESCetal} provides a suitable framework for the solution of the following problems:

\begin{itemize}
  \item Find a globally defined, \emph{smooth, time-invariant} state-feedback that achieves either one of the following asymptotic regulation objectives for nonholonomic systems: drive to \emph{zero} the state or drive the systems \emph{total energy} to a desired constant value.
\end{itemize}

The objectives should be ensured for initial conditions starting sufficiently close to the desired objective but outside a set which has zero Lebesgue measure and is nowhere dense.\footnote{Clearly, to comply with Brockett's necessary condition, in the case of regulation to zero this set should contain the origin.} Following standard practice, the qualifier \emph{``almost"} will be used to underscore the latter feature.

For the nonholonomic integrator we give an \emph{almost global} solution for both problems---that is, \emph{all} trajectories starting outside a zero-measure set converge to their desired value. For the general case of higher-order nonholonomic systems in chained form, it is shown that the EPD IDA-PBC matching equation is always solvable, and a local result is given.

\begin{figure}
    \centering
    \includegraphics[width=5cm]{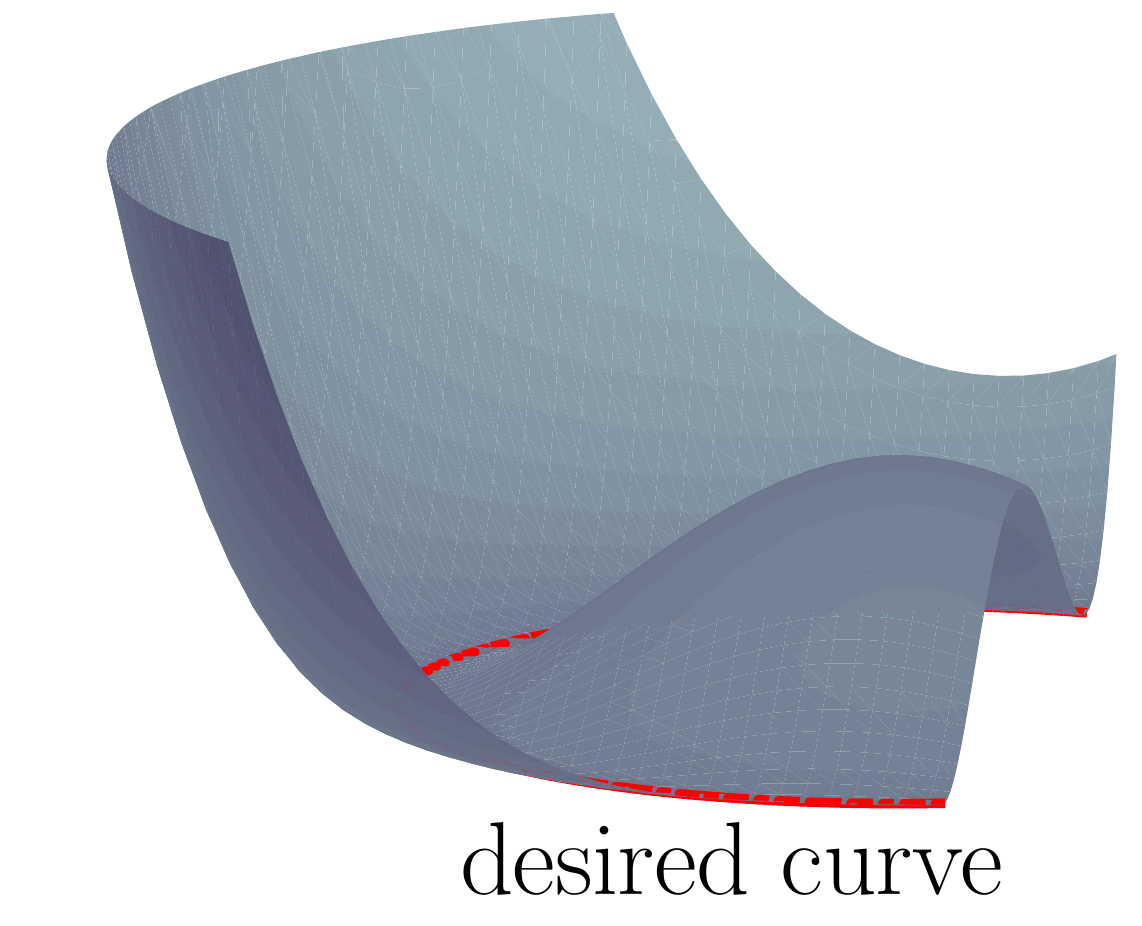}
    \caption{Shape of the energy function assigned by the FOC, with the desired periodic orbit in red}
    \label{figsom}
  \end{figure}

The reminder of the paper is organized as follows. In Section \ref{sec2} we introduce the problem formulation and the EPD IDA-PBC method to achieve \emph{almost global regulation} of nonholonomic systems in its general form. In Section \ref{sec3} we give the constructive solutions for the nonholonomic integrator, which are extended to high-order nonholonomic systems in chained form in Section \ref{sec4}. The paper is wrapped-up with simulations results in Section \ref{sec5} and concluding remarks in Section \ref{sec6}.\\

\textbf{Notation.} $I_n$ is the $n\times n$ identity matrix. For $x \in \rea^n$, $W \in \rea^{n \times n}$, $W=W^\top>0$, we denote the Euclidean norm $|x|^2:=x^\top x$, and the weighted-norm $\|x\|^2_W:=x^\top W x$. All mappings are assumed smooth. Given a function $H: \rea^n \to \rea$ we define the differential operator $\nabla H(x):=\left({\partial H\over \partial x}\right)^\top$.

\section{Regulation of Nonholonomic Systems via EPD IDA-PBC}
\lab{sec2}
%
In this paper, we adopt the driftless system representation of the nonholonomic system
\begin{equation}
\label{model_mech}
  \dot{x} = S(x) u,
\end{equation}
with $x \in \rea^n$ the generalized position, $u\in \rea^m$ the velocity vector, which is the control input, $n>m$, and the mapping $S: \rea^n \to \rea^{n\times m}$. The corresponding constraint is
\begin{equation}
\label{constraint}
  A^\top(x)\dot x = 0
\end{equation}
with $A: \rea^n \to \rea^{n \times (n-m)}$ full-rank. It is assumed that the system is completely nonholonomic, hence controllable. We refer the reader to \cite{Bloch2003} for further details on nonholonomic systems.

The proposition below shows that the problem of {regulation} of the system \eqref{model_mech} can be recast as an EPD IDA-PBC design. Following the ``FOC approach" advocated in \cite{ESCetal}---see also \cite{YIetal}---the idea is to decompose the state of the system into two components as\footnote{See Remark \ref{rem1} for the case of an arbitrary state partition.}
\begequ
\lab{parsta}
\begmat{ x_\ell \\ x_0}=x,\;x_\ell \in \rea^{n_\ell},\;x_0 \in \rea^{n_0}
\endequ
with $n=n_0+n_\ell$ and to find a smooth state-feedback that transforms the closed-loop dynamics into a pH system of the form
\begin{equation}
\label{cloloo}
 \begmat{ \dot x_\ell \\ \dot x_0} =  \begmat{ J_\ell(x) - R_\ell(x) & 0 \\ 0 & J_0(x) - R_0(x)} \nabla H(x),
\end{equation}
where the total energy function is given by
$$
H(x):= H_\ell(x_\ell) + H_0(x_0),
$$
with  $H_0:\rea^{n_0} \to \rea,\;H_\ell:\rea^{n_\ell} \to \rea$ and the interconnection and damping matrices
\begali{
\lab{intdammat}
&J_0:\rea^n \to \rea^{n_0 \times n_0},\;J_\ell:\rea^n \to \rea^{n_\ell \times n_\ell},\;R_0:\rea^{n} \to \rea^{n_0 \times n_0},\;R_\ell:\rea^{n} \to \rea^{n_\ell \times n_\ell}
}
satisfying
\begali{
\lab{conintdam}
&J_0(x)=-J_0^{\top}(x),\;J_\ell(x)=-J_\ell^{\top}(x),\;R_0(x)=R_0^{\top}(x) \geq 0.
}
The control objectives are to ensure that
\begequ
\lab{concon}
\lim_{t\to\infty} H_\ell(x_\ell(t)) = \beta_\ell>0, \quad \lim_{t\to\infty}x_0(t) =0,
\endequ
or
\begequ
\lab{concon0}
\lim_{t\to\infty} x(t) =0.
\endequ

As seen in the proposition below these objectives are achieved making the trajectory converge to the curve $H_\ell(x_\ell(t)) = \beta_\ell$ via the EPD principle, where $\beta_\ell>0$ in the first case and $\beta_\ell=0$ to regulate the state to zero. The EPD principle imposes the following constraint on $R_\ell(x)$:
  \begali{
  \lab{sigrl}
 [R_\ell(x) +  R^{\top}_\ell(x)] H^{\tt s}_\ell(x_\ell)  & \geq 0,
  }
where we defined the (shifted) energy function
\begequ
\lab{tilhl}
H^{\tt s}_\ell(x_\ell):= H_\ell(x_\ell) - \beta_\ell.
\endequ

To streamline the presentation of the result we partition the matrix $A(x)$ as
\begequ
\lab{a}
\begin{aligned}
& A(x)=\begmat{A_\ell(x) \\ A_0(x)},
\end{aligned}
\endequ
with $ A_\ell:\rea^n \to \rea^{n_\ell \times (n-m)}$ and $A_0:\rea^n \to \rea^{n_0 \times (n-m)}$.
%
\begin{proposition}
\rm\label{pro1}
Consider the system \eqref{model_mech} and the state partition \eqref{parsta}. Fix $\beta_\ell\ge0$. Assume there exist energy functions $H_0(x_0),\;H_\ell(x_\ell)$ and interconnection and damping matrices \eqref{intdammat}, \eqref{conintdam} verifying the following conditions.
\begenu
\item[{\bf C1.}] The matching PDE
\begin{equation}\label{pde_general}
 A_\ell^\top(x) \Big[J_\ell(x) - R_\ell(x) \Big] \nabla H_\ell(x_\ell)+ A_0^\top(x) \Big[J_0(x) - R_0(x) \Big] \nabla H_0(x_0)  =0.
\end{equation}
\item[{\bf C2.}] The EPD condition \eqref{sigrl}.
\item[{\bf C3.}] The minimum condition
\begequ
\lab{minconh}
{\nabla H(x)|_{x=0} = 0, ~\quad \nabla^2H(x) >0,}
\endequ
and the origin is isolated.
\item[{\bf C4.}] Define the function
\begequ
\lab{q}
Q(x):=\|\nabla H_0(x_0)\|^2_{R_0(x)} + \hal H^{\tt s}_\ell(x_\ell) \|\nabla H_\ell(x_\ell)\|^2_{R_\ell(x) +  R^{\top}_\ell(x)}.
\endequ
For the system \eqref{cloloo}, there exists a function $h:\rea^n \to \rea$ such that the following detectability-like implication holds
\begali{
\label{detcon}
\big[{ \lim_{t\to\infty} Q(x(t))= 0} \;\mbox{ and }\; x(0) \notin \cali\big] \; \Rightarrow \; \eqref{concon} \mbox{~~[{or \eqref{concon0}, respectively}]},
}
with $\cali:=\{x \in \rea^n\;|\;  h(x)= 0 \}$ the {\em inadmissible} initial condition set.
\endenu

Assume the initial conditions of the system are outside the set $\cali$. Then, the control law
\begin{equation}
\label{u}
  u = [S^\top  (x) S (x)]^{-1}S^\top  (x) \begmat{\big(J_\ell(x) - R_\ell(x) \big) \nabla H_\ell(x_\ell) \\ \big(J_0(x) - R_0(x) \big) \nabla H_0(x_0) }.
\end{equation}
ensures \eqref{concon} when $\beta_\ell>0$ or  \eqref{concon0} when $\beta_\ell=0$.
\end{proposition}
\begin{proof}\rm
Some simple calculations show that the closed-loop dynamics takes the pH form \eqref{cloloo}. Define the function
\begequ
\lab{v}
V(x) = {1\over2} |H^{\tt s}_\ell(x_\ell)|^2 + H_0(x_0),
\endequ
the derivative of which is
\begequ
\label{dotv}
\dot{V}  =  - Q(x)\leq 0,
\endequ
where the upperbound is obtained using \eqref{sigrl}. {Invoking the properties of the function $V(x)$, we conclude Lyapunov stability of the closed-loop dynamics with respect to the energy level set
$
\{x\in \rea^n| H_\ell(x_\ell) = \beta_\ell,~ x_0=0\}
$
for $\beta_\ell >0$ (or the origin if $\beta_\ell =0$, respectively).

According to Barbalat's Lemma, together with \eqref{dotv}, we have
$$
\lim_{t\to\infty} Q(x(t)) =0
$$
for any intial conditions. Since $x(0) \notin \cali$, and using the convergence implication \eqref{detcon} directly, we have that \eqref{concon} holds when $\beta_\ell>0$.} If $\beta_\ell=0$ we conclude \eqref{concon0} recalling the minimum condition {\bf C3}.
\end{proof}

 Proposition \ref{pro1} for $\beta_\ell=0$ does not contradict Brockett's necessary condition. Indeed, in the proposed design we only guarantee that the origin of the closed-loop system is Lyapunov stable but \emph{not asymptotically} stable. More precisely, we establish the following implication
      $$
      \Big[ ~| x(0) | < \delta, \; x(0) \notin \cali ~\Big]  \; \Rightarrow \; \lim_{t\to\infty} x(t)  =0,
      $$
which differs from the usual attractivity condition $| x(0) | < \delta \; \Rightarrow \; \lim_{t\to\infty} x(t)  =0$. See Fig. \ref{fig:ills}. Interestingly, we have the following lemma, whose proof is given in Appendix \ref{appa}.
\begin{lemma}
\lab{lem1}
Consider the scenario of Proposition 1, with $\beta_\ell=0$. The following implication is true: $$[{\bf C1-C4}\; \Rightarrow \; \{0\} \subset\;\cl(\cali)],$$
where $\cl(\cdot)$ denotes the closure of the set.
\end{lemma}

\begin{figure}[h]
    \centering
    \includegraphics[width=6cm]{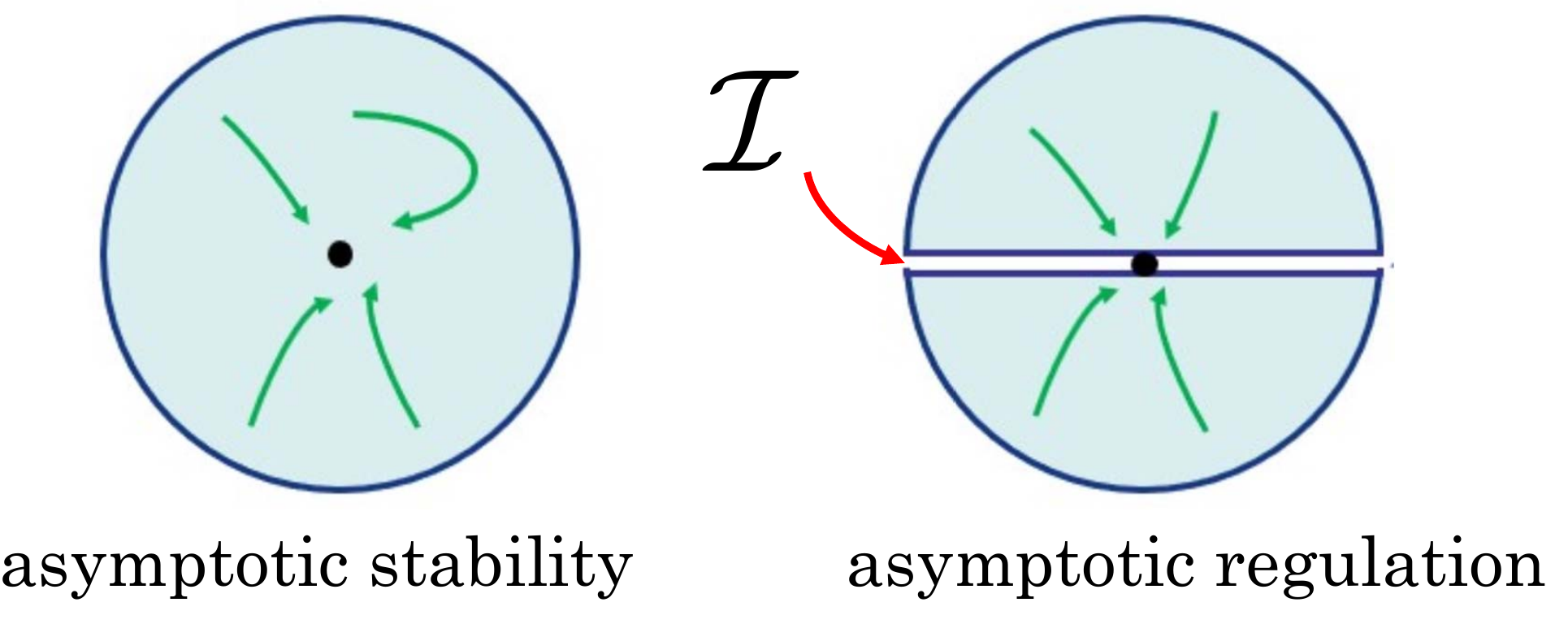}
    \caption{An illustration of the difference between asymptotic stability and the case studied in the paper}
    \label{fig:ills}
\end{figure}

\begrem
\lab{rem1}
To simplify the presentation we have assumed the direct partition of the state given in \eqref{parsta}. Proposition \ref{pro1} can be easily extended to the case where the partition is of the form
$$
\begmat{ x_\ell \\ x_0}=\calp x,
$$
where $\calp \in \rea^{n \times n}$ is a permutation matrix.
\endrem

\begrem
\lab{remr}
The EPD principle is codified in the inequality \eqref{sigrl} and graphically illustrated in Fig. \ref{fig_pd}. Clearly, when $\beta_\ell =0$, the EPD IDA-PBC becomes the standard IDA-PBC with damping injection ensured by \eqref{sigrl}.
\endrem

\begrem
\lab{rem4}
It should be pointed out that when $\beta_\ell>0$ the origin $x=0$ is an unstable equilibrium point of the closed-loop system. Indeed, the minimum condition \eqref{minconh} ensures that $H(x)$ is a (locally) positive definite function whose derivative is given as
$$
\dot H=-\|\nabla H_0(x_0)\|^2_{R_0(x)} - \|\nabla H_\ell(x_\ell)\|^2_{R_\ell(x) +  R^{\top}_\ell(x)}.
$$
On the other hand, in a small (relative to $\beta_\ell$) neighborhood of $x_\ell=0$, the EPD condition \eqref{sigrl} imposes that $R_\ell(x) +  R^{\top}_\ell(x)<0$. Hence, there exists a neighborhood of $x=0$ where $\dot H>0$ that---according to Lyapunov's first instability theorem \cite{KHAbook}---implies that the origin is unstable.
\endrem

\begin{figure}
    \centering
    \includegraphics[width=7cm]{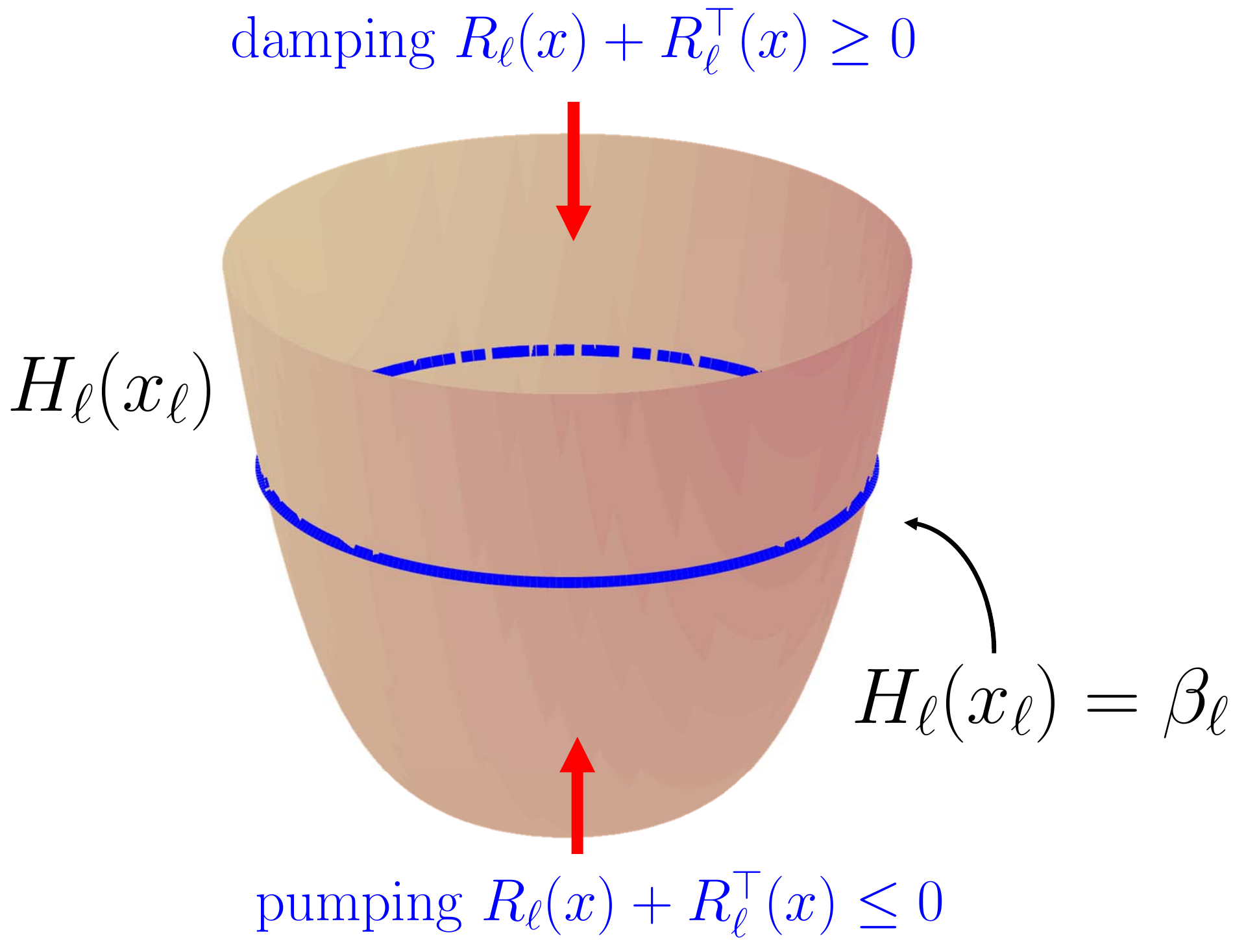}
    \caption{An interpretation to the EPD control method}
    \label{fig_pd}
  \end{figure}

\begrem
Standard IDA-PBC has been applied in \cite{DONetal,MASSCH} to stabilise a manifold containing the desired equilibrium point, in the latter publication including disturbance rejection. In \cite{FERetal,FUJetal} switched or non-smooth versions of IDA-PBC that ensure convergence to the desired equilibrium point are proposed.
\endrem

\begrem
As indicated in the Introduction, the (mathemathically elegant) transverse function method of \cite{MORSAMtac,MORSAM} follows the same approach adopted here---which was originally inspired by \cite{ESCetal}.\footnote{The first use of the FOC approach for the control of nonholonomic systems is, erroneously credited to \cite{DIXetal} in  \cite{MORSAMtac,MORSAM}. In view of the tangential reference to \cite{ESCetal} made in \cite{DIXetal}, this is probably inadvertendly.} This method can be used for the tracking problem of controllable driftless systems invariant on a Lie group.
\endrem
%
\section{Nonholonomic Integrator}
\lab{sec3}
%
In this section, we consider the benchmark example of the nonholonomic integrator described in chained form by
\begali{
\nonumber
\dot x_1 & = u_1\\
\nonumber
\dot x_2 & = u_2\\
\lab{nonholint}
\dot x_3 &= x_2 u_1.
}
The system can be represented in the form \eqref{model_mech}, \eqref{constraint} with the definitions
$$
S(x)=\begmat{1 & 0 \\ 0 & 1 \\ x_2 & 0},\;A(x)=\begmat{x_2 \\ 0 \\ -1}.
$$

The proposition below solves, via direct application of Proposition \ref{pro1}, the problems of regulation of the energy or driving the state to zero for the system \eqref{nonholint}.
%
\begin{proposition}
\label{pro2}\rm
Consider the nonholonomic system \eqref{nonholint} with the state partition $x_\ell=\col(x_1,x_2)$ and $x_0=x_3$. Fix  $\beta_\ell \geq 0$.
\begenu
\item[{\bf P1.}] The functions
$$
H_\ell(x_\ell)=\hal(x_1^2+x_2^2),\;\;H_0(x_0)=\hal x_3^2,
$$
together with the mappings
\begalis{
&J_0=0,\; R_0(x)=x_2^2,\; J_\ell(x)=\begmat{0 & -x_3 \\ x_3 & 0}, \;R_\ell(x)=\begmat{0 & 0 \\ 0 & \gamma H^{\tt s}_\ell(x_\ell)},
}
with $\gamma>0$ and $H^{\tt s}_\ell(x_\ell)$ defined in \eqref{tilhl}, verify conditions {\bf C1}-{\bf C3} of Proposition \ref{pro1}.
\item[{\bf P2.}] The smooth, time-invariant control law \eqref{u} takes the form\footnote{We have splitted the control law into $\vida(x)$ and $\vpd(x)$ to underscore the role of energy-shaping and EPD terms, respectively.}
\begequ
\lab{u1}
u =  \vida(x) + \vpd(x)
\endequ
with
\begalis{
\vida(x) & = \begmat{  -x_2 x_3 \\ x_1x_3  }\\
\vpd(x) & = \begmat{  0 \\ - \gamma H^{\tt s}_\ell(x_1,x_2) x_2 }.
}
\item[{\bf P3.}]   The function $Q(x)$, defined in \eqref{q}, is given as
\begequ
\lab{qexa}
Q(x)= - x_2^2 \big[\gamma (x_1^2+x_2^2-\beta_\ell)^2 + x_3^2\big].
\endequ
\item[{\bf P4.}] The control law \eqref{u1} ensures \eqref{concon} when $\beta_\ell>0$ or \eqref{concon0} when $\beta_\ell=0$ and $\gamma=1$ with the set of inadmissible initial conditions in {\bf C4} defined via the function
\begequ
\lab{hpro2}
h(x)=(x_1^2+x_2^2)x_3^2.
\endequ
\endenu
\end{proposition}

\begin{proof}\rm
The proof of claims {\bf P1}-{\bf P3} follows via direct calculations, noting that the closed-loop system takes the pH form
\begin{equation}
\label{cloloopro2}
  \dot x =
  \begin{bmatrix}
    0 & -x_3 & 0 \\ x_3 & - \gamma H^{\tt s}_\ell(x_\ell)  & 0 \\ 0 & 0 & - x_2^2
  \end{bmatrix}
  \begin{bmatrix}
  \nabla H_\ell \\ \nabla H_0
  \end{bmatrix}.
\end{equation}

To apply Proposition \ref{pro1} we need to prove the detectability-like condition \eqref{detcon} with the function $h(x)$ given in \eqref{hpro2}. First, we note that $\{x\in\rea^3| x_1=x_2=0 ~\mbox{or}~x_2=x_3=0\} \subset \cali$ is an equilibrium set for the closed-loop system \eqref{cloloopro2} that does not match the control objectives---therefore it has to be ruled out. We will now prove that $\cali$ is the set of inadmissible initial conditions.

Starting outside $\cali$, La Salle's Invariance Principle ensures that all trajectories will converge to the largest invariant set contained in the set $\{x \in \rea^n\;|\;Q(x)=0\}$. {Thus, we have the systems state converges to the set $\{x\in\rea^3|H_\ell(x_\ell) = \beta_\ell,~x_3=0\}$ or the largest invariant set in $\{x\in\rea^3|x_2=0\}$. Since the former is exactly the desired task,} it is clear that we only need to prove the implication \eqref{detcon} for the case {$\lim_{t\to\infty}x_2(t)=0$}. Towards this end, we first note that
\begin{equation}
\begin{aligned}
\label{dottilhell}
  \dot{H^{\tt s}_\ell} & = \dot H_\ell \\
  & = (\nabla H_\ell)^\top
  \begin{bmatrix}
    0 & -x_3 \\ x_3 & - \gamma H^{\tt s}_\ell
  \end{bmatrix}
  \nabla H_\ell \\
  & = - \gamma x_2^2 H^{\tt s}_\ell.
\end{aligned}
\end{equation}
Similarly, we have from \eqref{cloloopro2} that
$$
\dot x_3  = - x_2^2 x_3.
$$
From these two equations we conclude that
\begali{
\lab{limx3}
\lim_{t\to\infty} x_3(t) =0 & \; \Leftrightarrow \; \lim_{t\to\infty} H^{\tt s}_\ell(x_\ell(t)) =0, \;  \Leftrightarrow \; x_2(t) \notin \call_2.
}
Now, solving \eqref{dottilhell} we get
$$
{H}_\ell(x_\ell(t))  = e^{-\gamma \int_0^t x_2^2(s)ds} H_\ell(x_\ell(0))+ \beta_\ell[1- e^{-\gamma \int_0^t x_2^2(s)ds}].
$$
In view of the constraint on the initial conditions, \emph{i.e.}, $H_\ell(x_\ell(0)) \neq 0$, and the fact that $\beta_\ell \geq 0$, we have that
\begequ
\lab{hlpos0}
H_\ell(x_\ell(t))=\hal[x_1^2(t)+x_2^2(t)]>0,\;\forall t\in [0,\infty).
\endequ
Moreover, if $\beta_\ell > 0$ we also have that
\begequ
\label{hlpos-lim}
\lim_{t\to\infty} H_\ell(x_\ell(t)) >0.
\endequ
The equivalences \eqref{limx3} and the inequalities \eqref{hlpos0},  \eqref{hlpos-lim} will be instrumental to complete the proof.

{To apply LaSalle Invariance Principle, let us calculate the largest invariant set in $\{x\in \rea^3| x_2=0\}$.} The second equation in \eqref{cloloopro2} is given by
$$
\dot x_2=x_1 x_3-{\gamma \over 2}x_2(x_1^2+x_2^2-\beta_\ell)
$$
from which we conclude that
\begalis{
x_2 \equiv 0 & \;\Rightarrow\; x_1x_3 \equiv 0.
}
{It implies three cases: 1) $x_1 \equiv 0$, 2) $x_3 \equiv 0$, or 3) non-zero signals $x_1(t)$ and $x_3(t)$ are orthogonal, {\em i.e.}, $x_1(t)x_3(t)=0$. From the dynamics of $x_3$, we get the monotonicity of $x_3$ with respect to time, thus excluding the third case. It yields
$$
\lim_{t\to\infty} x_1(t) =0\quad \mbox{or} \quad
\lim_{t\to\infty} x_3(t) =0.
$$
}
We proceed now to prove that the latter implies \eqref{concon}. {If the systems state converges to the invariant set $\{x_2=x_3=0\}$}, we conclude that the trajectories of the closed-loop system verify $\lim_{t\to\infty} x_3(t) =0$, but from \eqref{limx3} we have that this is possible if and only if $\lim_{t\to\infty} H^{\tt s}_\ell(x_\ell(t)) =0$. Therefore, we only need to consider the case of {convergence to the invariance set $\{x_1=x_2=0\}$, {\em i.e.},} {$\lim_{t\to\infty}(x_1(t),x_2(t))=0$.}

If $\beta_\ell>0$ the inequalities \eqref{hlpos0} and \eqref{hlpos-lim} rule-out the possibility of {$\lim_{t\to\infty}x_1(t)=\lim_{t\to\infty}x_2(t)=0$,} completing the proof for this case.

Let us consider now the case $\beta_\ell=0$. In this case, we have that the function $V(x)$---defined in \eqref{v}---takes the form
$$
V(x) = {1\over2} H_\ell^2(x_\ell) + H_0(x_0),
$$
and its derivative is given by
$$
\dot{V}  =  - 2x_2^2 V.
$$
Consequently, recalling \eqref{limx3}, we have that
\begequ
\label{limv}
\lim_{t\to\infty} V(x(t)) =0 \; \Leftrightarrow\; \lim_{t\to\infty} H_\ell(x_\ell(t)) =0.
\endequ
The proof is concluded noting that if {$\lim_{t\to\infty}(x_1(t),x_2(t))=0$}, the trajectories of the closed-loop system verify $\lim_{t\to\infty} H_\ell(x_\ell(t)) =0$.
\end{proof}

\begrem
The controller of Proposition \ref{pro2} (with $\beta_\ell=0$ and $\gamma=1$) solves the problem of almost global regulation to zero of the nonholonomic integrator with a \emph{smooth, time-invariant} state-feedback. To the best of our knowledge, such a problem was still open in literature.
\endrem

\begrem
Although not necessary for the analysis of the asymptotic behavior in Proposition \ref{pro2}, we have added in the control a tunable parameter $\gamma>0$ that, as shown in the simulations, enhances the performance. For the case of regulation of the state to zero, this parameter is taken equal to one. However, it is possible to add this tuning gain in an alternative controller, which incorporates a dynamic extension that makes the constant $\beta_\ell$ a function of time  $\beta_\ell(t)$ that asymptotically converges to zero.
\endrem
\begrem
Due to its smoothness and time-invariance, it is reasonable to expect that the transient performance of the proposed design is better than the one resulting from the application of time-varying \cite{POM,JIAetal}, discontinuous  \cite{AST,FUJetal} or switching \cite{LIB} feedback laws. {Specifically, for the former there is poor transient performance with the presence of oscillations, which is clearly caused by injecting sinusoidal signals. For the discontinuous feedback, see \cite{AST} for instance, the sensitivity to measurement noise is due to the appearance of some state in the denominator.} This fact is illustrated via simulations in Section \ref{sec5}.
\endrem
\begin{remark}
{Notice that the system \eqref{nonholint} is diffeomorphic to the system considered in \cite{ESCetal}, that is,
\begalis{
\dot z_1 & = u_1\\
\dot z_2 & = u_2\\
\dot z_3 &= z_1 u_2 - z_2 u_1,
}
via the change of coordinates $z \mapsto (x_1,x_2,x_1x_2-2x_3)$, and they are both particular cases of the dynamical model of the current-fed induction motor \cite{ESCetal,MARbook}.}
\end{remark}

\section{Nonholonomic Systems in Chained Form}
\label{sec4}

Now we extend the results to the high dimensional nonholonomic systems with chained structure. That is, the $n$-dimensional system \eqref{model_mech} with
\begin{equation}
\label{S}
S(x) =
\begin{bmatrix}
  1 & 0 \\ 0 & 1\\ x_2 & 0 \\ x_3 & 0 \\ \vdots & \vdots \\ x_{n-1} & 0
\end{bmatrix}.
\end{equation}
It is well-known \cite{MURSAS,ORI} that arbitrary nonholonomic systems of order $n \leq 4$ can always be transformed into the previous chained form. Hence, the class considered in this section covers a large number of practical applications.

We have the following proposition whose proof is very similar to the proof of Proposition \ref{pro2}. Unfortunately, due to the complicated nature of the zero dynamics for the output $Q(x)$, we can only prove a \emph{local} convergence result for this general case.

\begin{proposition}
\label{prop4}\rm
Consider the nonholonomic system \eqref{model_mech}, \eqref{S} with the state partition $x_\ell=\col(x_1,x_2,x_4,\ldots,x_n)$ and $x_0=x_3$. Fix $\beta_\ell \geq 0$.
\begenu
\item[{\bf S1.}]  The functions
$$
H_0(x_0) = {1\over2} x_3^2,\;
H_\ell(x_\ell) = {1\over2} |x_\ell|^2,
$$
together with $J_0  =0$, $R_0(x) = x_2^2$ and the matrices
$$
\begin{aligned}
J_\ell(x) & =
\begin{bmatrix}
  0 & - x_3 & 0 & \ldots & 0 \\
  x_3 & 0  & x_3^2 & \ldots & x_3x_{n-1}  \\
  0  & 0  & 0 & \ldots & 0 \\
  0 & - x_3^2  & 0 & \ldots & 0 \\
  \vdots & \vdots  & \vdots & \vdots & \vdots \\
  0 & - x_3x_{n-1}  & 0 & \ldots & 0
\end{bmatrix}
\\
R_\ell(x) & =
\begin{bmatrix}
  0 & 0 & 0 & \ldots & 0 \\
  0 & \gamma H^{\tt s}_\ell(x_\ell)  & 0 & \ldots & 0  \\
  0  & 0  & 0 & \ldots & 0 \\
  0 & 0  & 0 & \ldots & 0 \\
  \vdots & \vdots  & \vdots & \vdots& \vdots \\
  0 & 0  & 0 & \ldots & 0
\end{bmatrix}
\end{aligned}
$$
where $H^{\tt s}_\ell(x_\ell)$ is defined in \eqref{tilhl} and $\gamma>0$, verify conditions {\bf C1}-{\bf C3} of Proposition \ref{pro1}.
\item[{\bf S2.}] The smooth, time-invariant control law \eqref{u} takes the form \eqref{u1} with
\begalis{
\vida(x) & =\begmat{  -x_2 x_3 \\ x_1x_3 + x_3\Big(x_3x_4 + \ldots + x_{n-1}x_n\Big) }\\
\vpd(x)& = \begmat{  0 \\ - \gamma H^{\tt s}_\ell(x_\ell) x_2 }.
}
\item[{\bf S3.}] The function $Q(x)$, defined in \eqref{q}, is given as
\begequ
\lab{qexa_prop3}
Q(x)= - x_2^2 \big(\gamma (H^{\tt s}_\ell(x_\ell))^2 + x_3^2\big).
\endequ
\item[{\bf S4.}] There exists $\delta_{\tt min}>0$ such that for all $\delta \leq \delta_{\tt min}$ the control law \eqref{u1} ensures \eqref{concon} when $\beta_\ell>0$ or \eqref{concon0} when $\beta_\ell=0$ and $\gamma=1$, {\em or} convergence to the following invariant set
$$
\{ x \in \rea^n ~| ~  x_2=0, x_1+ x_3x_4 + \ldots + x_{n-1}x_n =0,  x_3 \neq 0, H^{\tt s}_\ell(x_\ell)\neq 0 \},
$$
provided the initial state starts in the set
$
\{ x \in \rea^n ~| (H^{\tt s}_\ell(x_\ell))^2 + x_3^2 \leq \delta\}.
$
\endenu
\end{proposition}

\begin{proof}\rm
The proof of claims {\bf S1}-{\bf S3} follows via direct calculations, noting that the closed-loop system takes the pH form
\begin{equation}
\label{closed-loop2}
  \dot x =
\begin{bmatrix}
  0 & - x_3 & 0 & 0 & \ldots & 0 \\
  x_3 & -\gamma H^{\tt s}_\ell(x_\ell) & 0 & x_3^2 & \ldots & x_3x_{n-1}  \\
  0  & 0 & - x_2^2 & 0 & \ldots & 0 \\
  0 & - x_3^2 & 0 & \ldots & \ldots & 0 \\
  \vdots & \vdots & \vdots & \ldots &  \\
  0 & - x_3x_{n-1} & 0 & 0 & \ldots & 0
\end{bmatrix}
  \nabla H(x).
\end{equation}
Similarly to the case of the nonholonomic integrator, we also have the key relationships
$$
  \dot H^{\tt s}_\ell    = - \gamma x_2^2 H^{\tt s}_\ell,\;  \dot x_3    = - x_2^2 x_3.
$$
Hence, the equivalence \eqref{limx3} holds true. Also, in view of \eqref{qexa_prop3} we only need to study the case $x_2 \equiv 0$, when we have from the closed-loop dynamics \eqref{closed-loop2} that
$$
x_3(x_1 + x_3x_4 + \ldots + x_{n-1}x_n) =0.
$$
Hence, $x_3=0$ or $x_1 + x_3x_4 + \ldots + x_{n-1}x_n = 0$. In the first case, we clearly have $\lim_{t\to\infty}x_3 =0$, and using \eqref{limx3}, we conclude that $\lim_{t\to\infty}H^{\tt s}_\ell(x_\ell(t)))  =0$, achieving the control objective. Therefore, we conclude the state will converge into the following set
\begequ
\lab{resset}
\{ x \in \rea^n ~| ~H^{\tt s}_\ell(x_\ell)=0,\; x_3 =0 \}
\cup
\{ x \in \rea^n ~| ~  x_2=0, x_1+ x_3x_4 + \ldots + x_{n-1}x_n =0 \}.
\endequ
completing the proof.
\end{proof}

\begrem
\lab{rem9}
As shown in the proposition above, the matching PDEs are always solvable satisfying all the assumptions of the EPD IDA-PBC design. Unfortunately, for $n \geq 4$ the invariant set to which all trajectories converge given in \eqref{resset} contains, besides the target set, an additional set that complicates the convergence analysis. Thus, we can only guarantee local convergence. Moreover, simulation evidence proves that---starting far away from the desired equilibrium---the closed-loop system trajectories will not converge to their desired values, confirming the local nature of our result.
\endrem

\section{Simulations}
\label{sec5}

The performance of the proposed controller is illustrated via simulations with Matlab/Simulink, which are summarized as follows.

\begenu
\setlength{\itemsep}{7pt}
\item[{\bf E1}] In Fig. \ref{fig:simulation1} we give the simulation results of the \emph{energy} regulation controller of Proposition \ref{pro2}, \emph{i.e.}, with $\beta_\ell>0$, the initial conditions $x(0)=(3,2,2)^\top$ and $\gamma=5$.
\item[{\bf E2}] In Fig. \ref{fig:simulation2} we repeat the simulation above, but for \emph{state} regulation, that is, $\beta_\ell=0$. We also give the simulation results of the well-known Pomet's method \cite{POM} with
    $$
    \begin{aligned}
     u_1 & = - (x_2 + x_3\cos (t))x_2\cos (t) - (x_2x_3+x_1) \\
     u_2 & = x_3\sin (t) - (x_2 + x_1\cos (t)).
    \end{aligned}
    $$
As expected, due to the periodic signal injection in the feedback law---which was designed following the procedure proposed in \cite{POM}---large oscillations are observed in the lengthy transient stage. Clearly, the new design outperforms Pomet's method with a significantly improved transient performance.
\item[{\bf E3}] To evaluate the robustness of the EPD IDA-PBC method, we repeat the experiment above adding (unavoidable) high-frequency noise in the measurable state.\footnote{The measurement noise is generated by the block ``Uniform Random Number'', where the signals are limited to $[-0.1,0.1]$ and the sample times are selected as 0.01.} Fig. \ref{fig:simulation3} illustrates that the state now converge to a small neighborhood of the desired equilibrium point. Here, we compare the new design with the famous (exponentially convergent) discontinuous design of Astolfi \cite{AST}
    $$
    \begin{aligned}
     u_1 & = - kx_1 \\
     u_2 & = p_2 x_2 + p_3 {x_3\over x_1},
    \end{aligned}
    $$
with $k=1, p_2=-5$ and $p_2=9$. As shown in the figure the state trajectories grow unbounded in finite time. Thus \emph{ad-hoc} modifications are needed to deal with this problem in practice.
\item[{\bf E4}] Simulations of the energy regulation controller for the case $n=4$ were also carried-out. In Figs. \ref{fig:4-1}-\ref{fig:4-2}, we fix $\gamma=0.5$. The trajectory in Fig. \ref{fig:4-1} achieves the desired objective. However, Fig. \ref{fig:4-2} shows that it fails for a larger $x_4(0)$. If we fix $x_1(0)=0.5, x_2(0)=1$ and $x_3(0)=0.1$, after extensive simulations we can find the critical initial value for $x_4(0)$ to be between $0.9$ and $1$---for $x_4(0)>1$ the state will converge to the undesired set and the desired objective is not achieved. Increasing the parameter $\gamma$, as shown in Figs. \ref{fig:4-3}-\ref{fig:4-4}, the controller will achieve the desired target again. Roughly speaking, in this simulation case a larger $\gamma$ enlarges the ``domain of attraction''---however, this pattern was not observed in other simulation scenarios.
\endenu

\begin{figure}[h]
\centering
\subfloat[Trajectories in the state space]{
    \includegraphics[width=5cm,height=4cm]{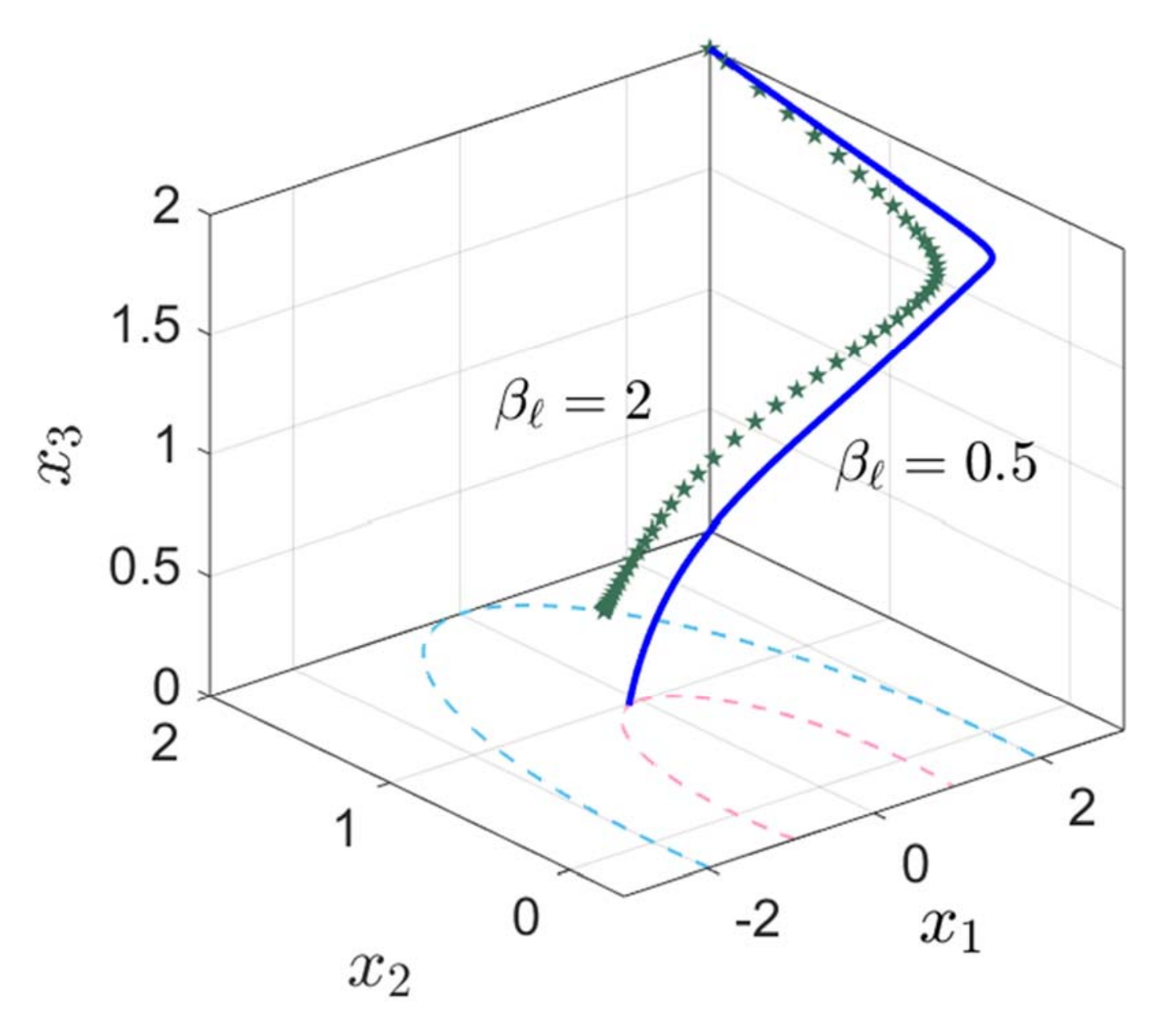}
    \label{fig:1-1}
}
\subfloat[$x_3(t)$ and partial energy $H_\ell(x_\ell(t))$]{
    \includegraphics[width=5cm,height=4cm]{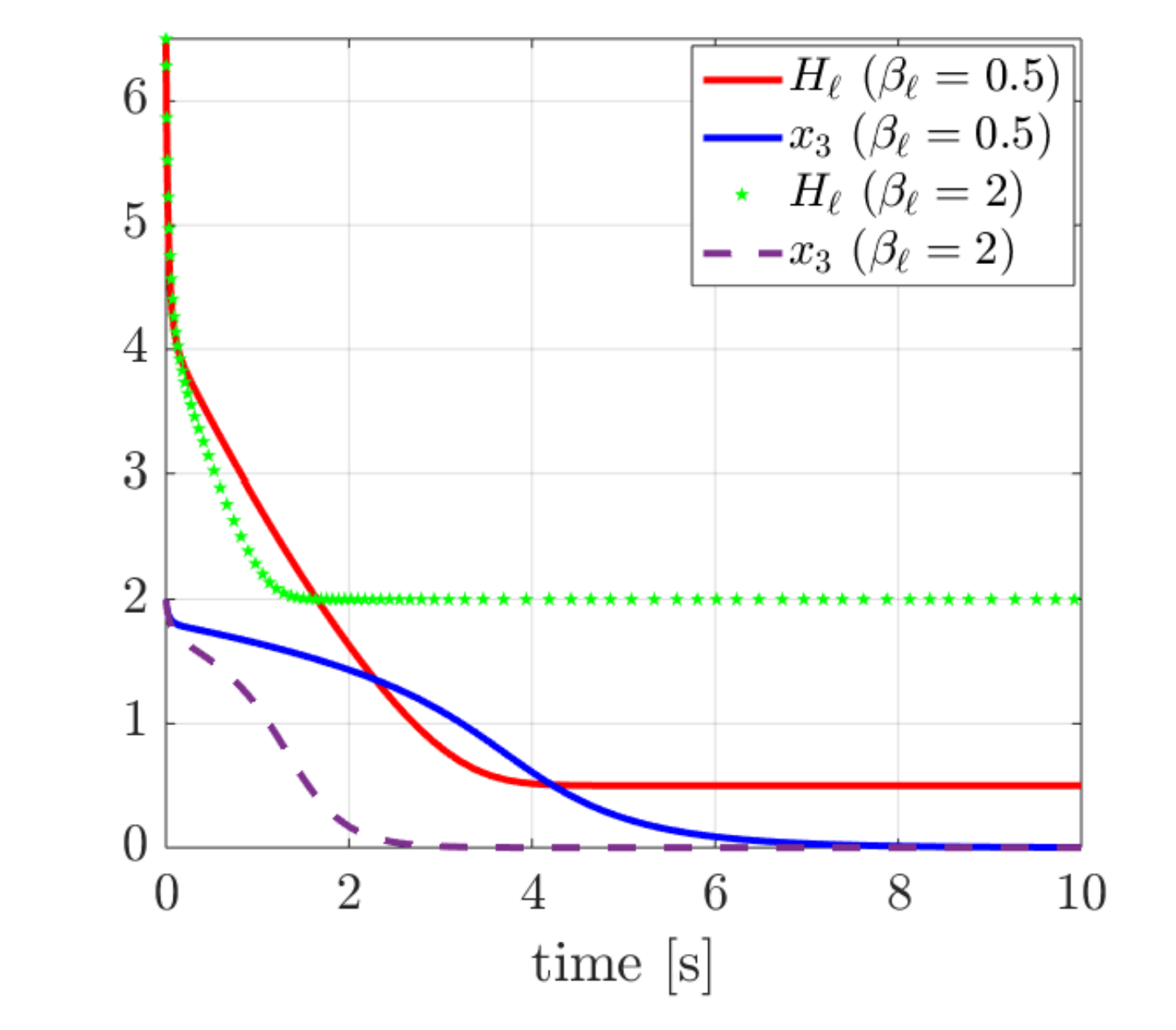}
    \label{fig:1-2}
}
\\
\subfloat[Trajectories of $x_\ell$ and function $H_\ell(x_\ell)$]{
    \includegraphics[width=5cm,height=4cm]{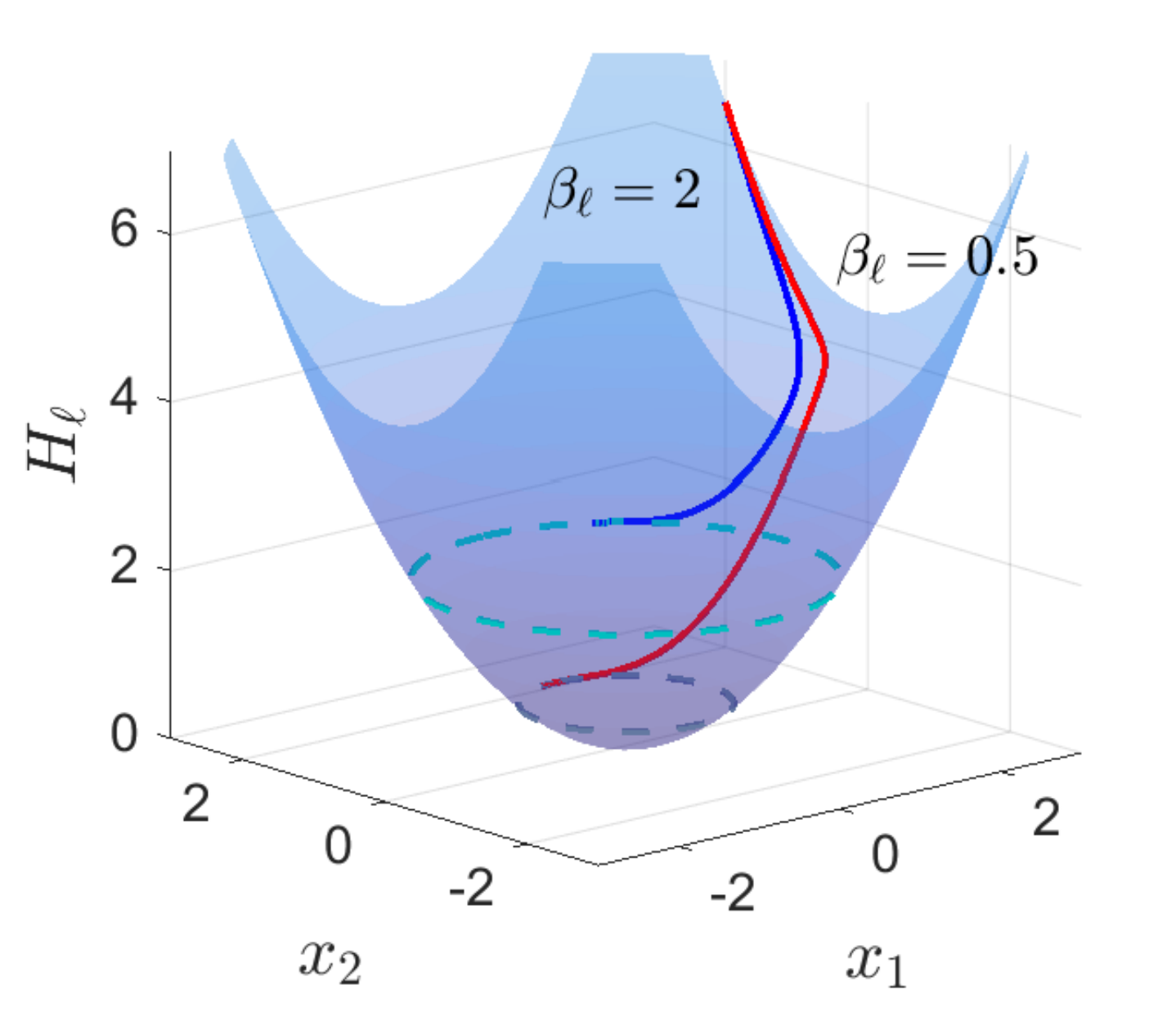}
    \label{fig:1-3}
}
\subfloat[Control inputs]{
    \includegraphics[width=5cm,height=4cm]{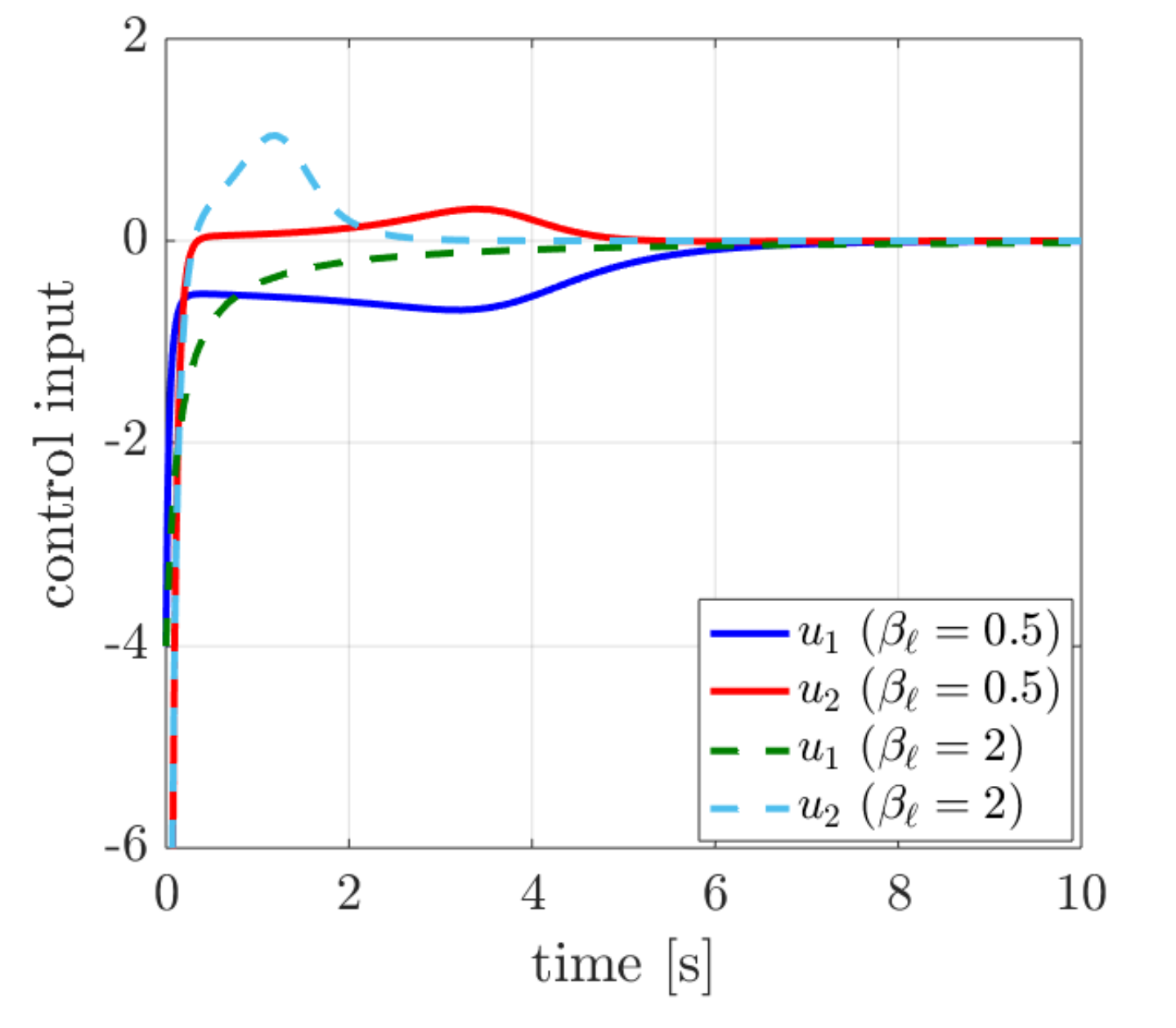}
    \label{fig:1-4}
}
\caption{Energy regulation (with $\beta_\ell=0.5,2$) of the nonholonomic integrator}
\label{fig:simulation1}
\end{figure}


\begin{figure}[h]
\centering
\subfloat[Trajectories in the state space]{
    \includegraphics[width=5cm,height=4cm]{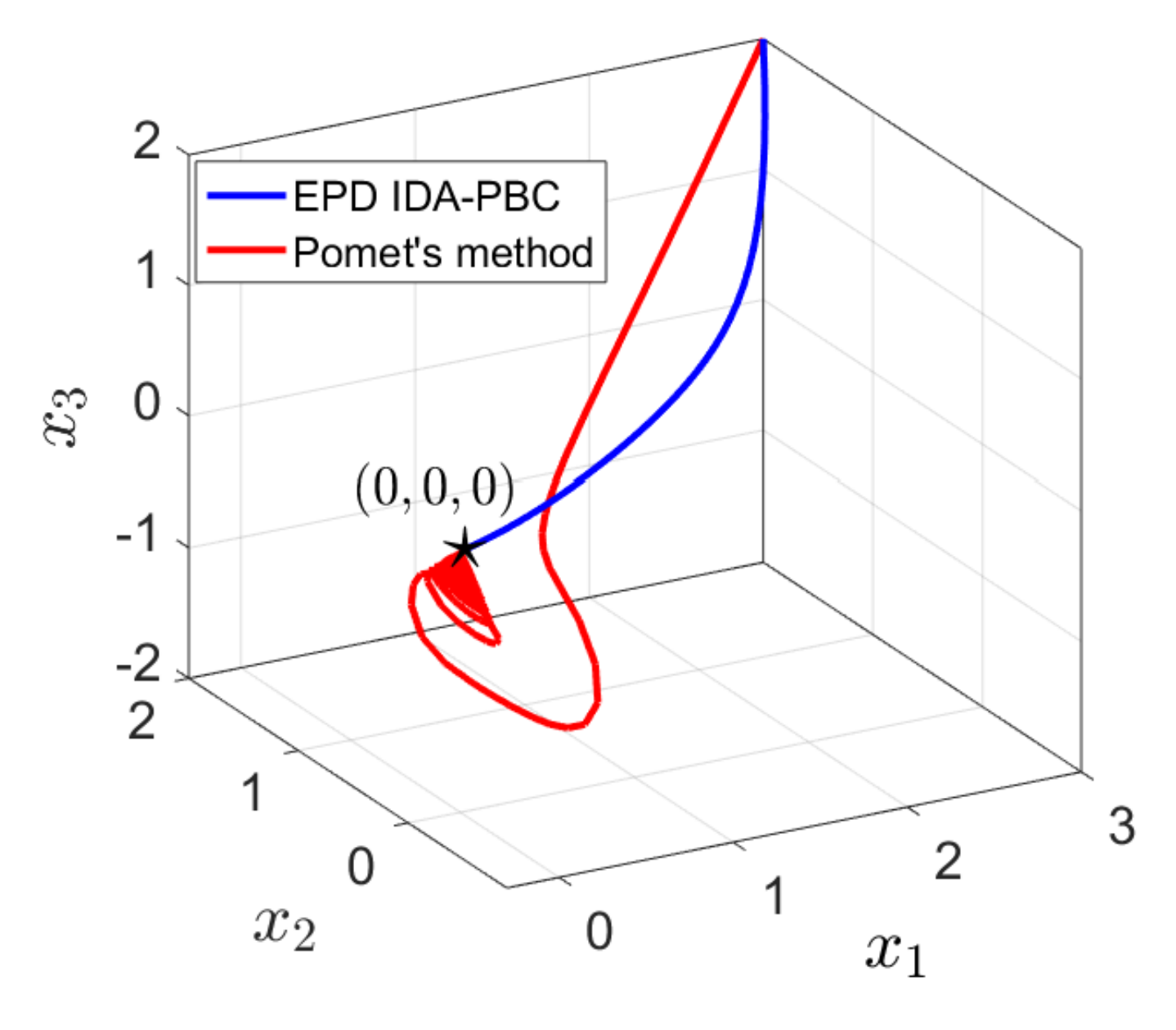}
    \label{fig:2-1}
}
\subfloat[State trajectories  $x(t)$ of Pomet's controller]{
    \includegraphics[width=5cm,height=4cm]{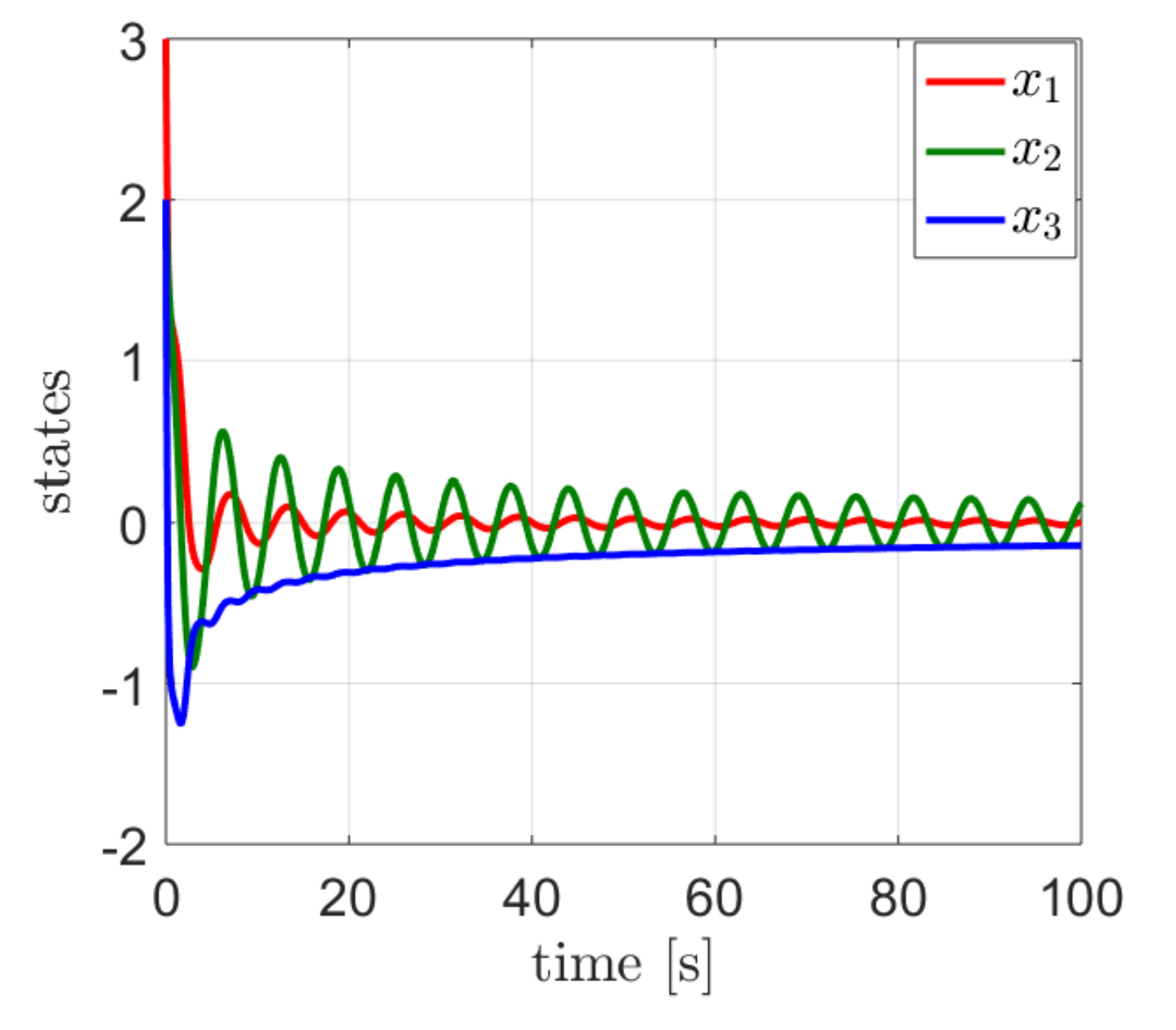}
    \label{fig:2-2}
}
\\
\subfloat[State trajectories  $x(t)$ of EPD IDA-PBC]{
    \includegraphics[width=5cm,height=4cm]{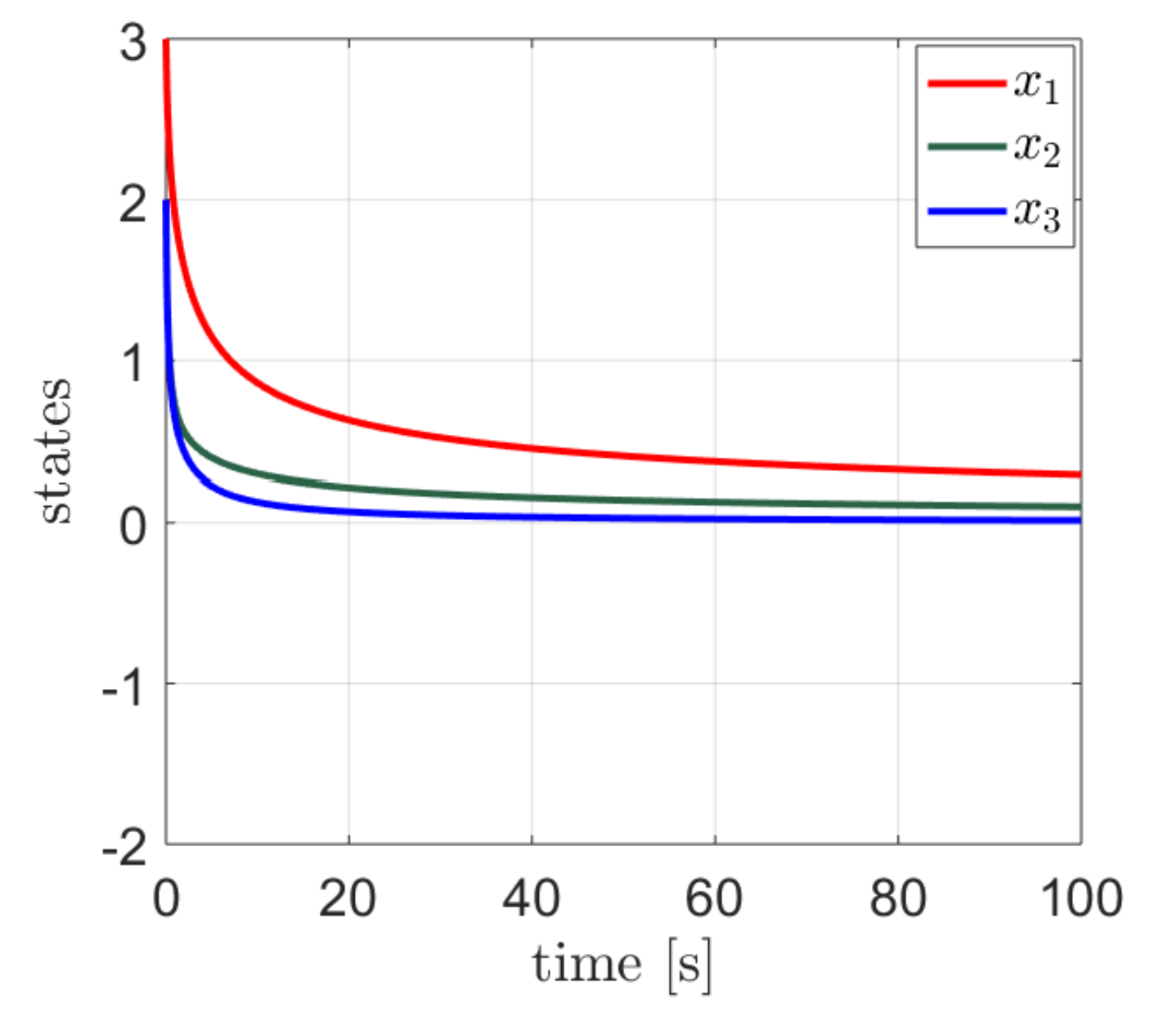}
    \label{fig:2-3}
}
\subfloat[Control inputs]{
    \includegraphics[width=5cm,height=4cm]{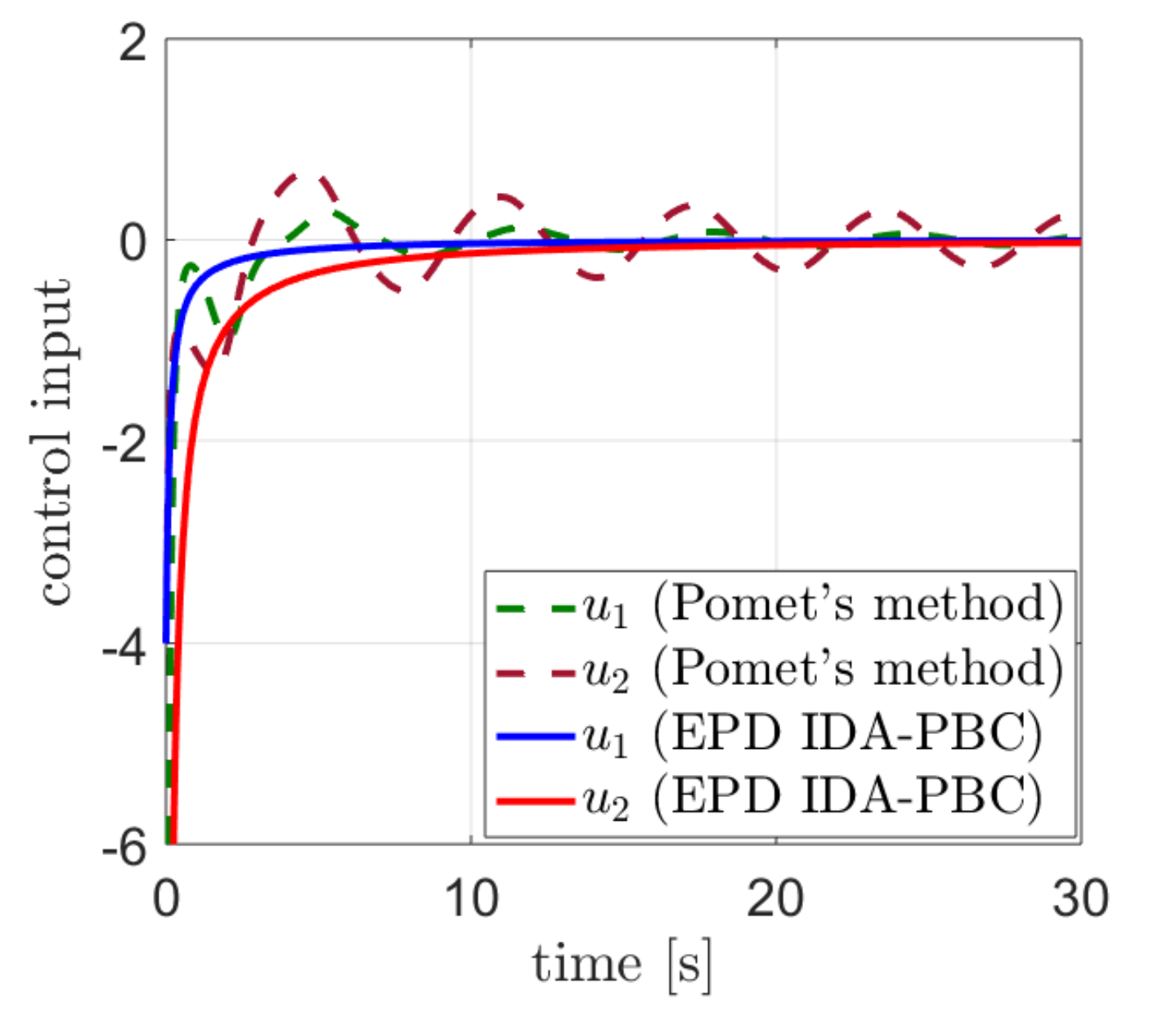}
    \label{fig:2-4}
}
\caption{State regulation of nonholonomic integrator with EPD IDA-PBC and Pomet's controller}
\label{fig:simulation2}
\end{figure}


\begin{figure}[h]
\centering
\subfloat[Trajectories in the state space]{
    \includegraphics[width=5cm,height=4cm]{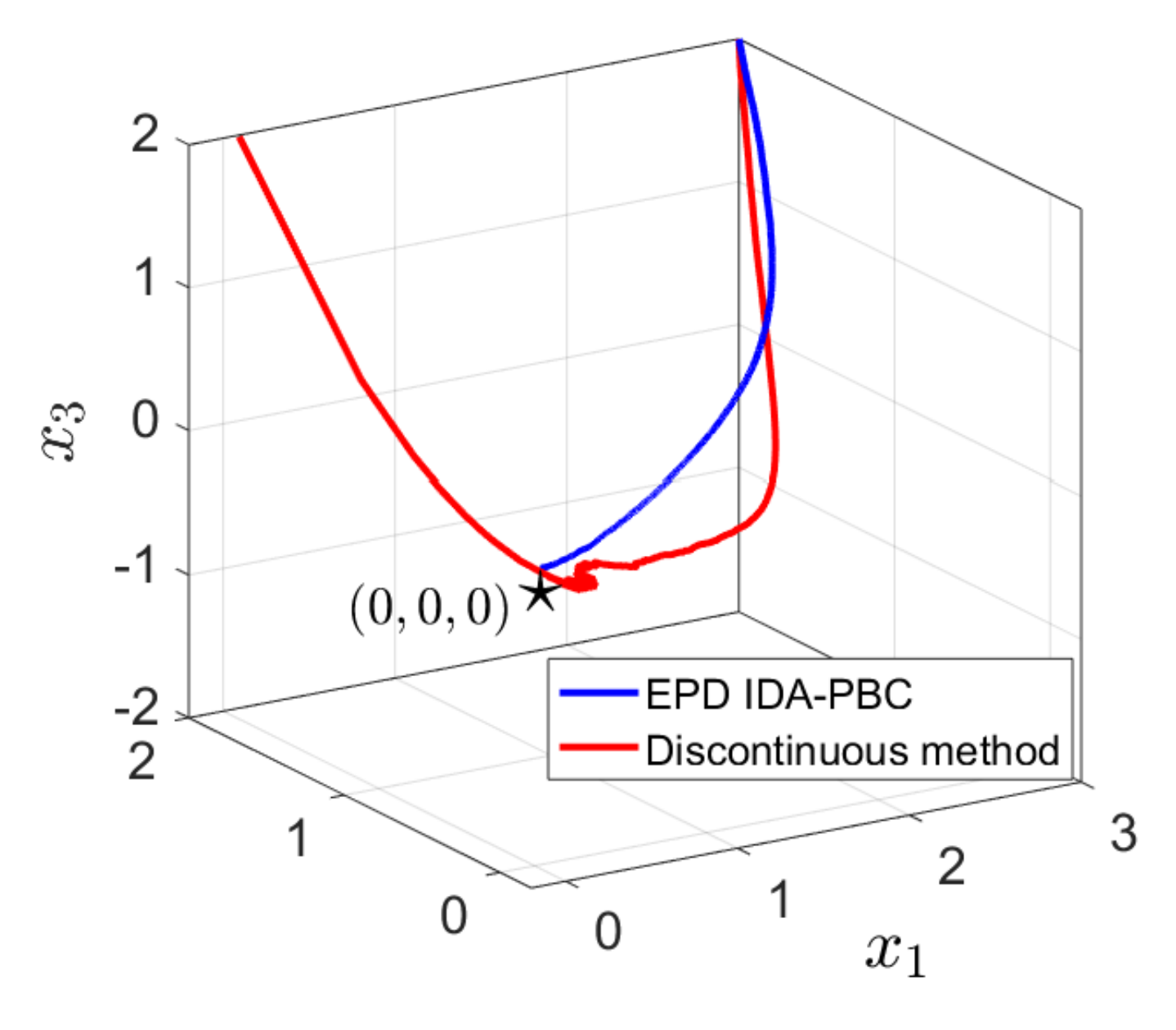}
    \label{fig:3-1}
}
\subfloat[State trajectories $x(t)$ of EPD IDA-PBC]{
    \includegraphics[width=5cm,height=4cm]{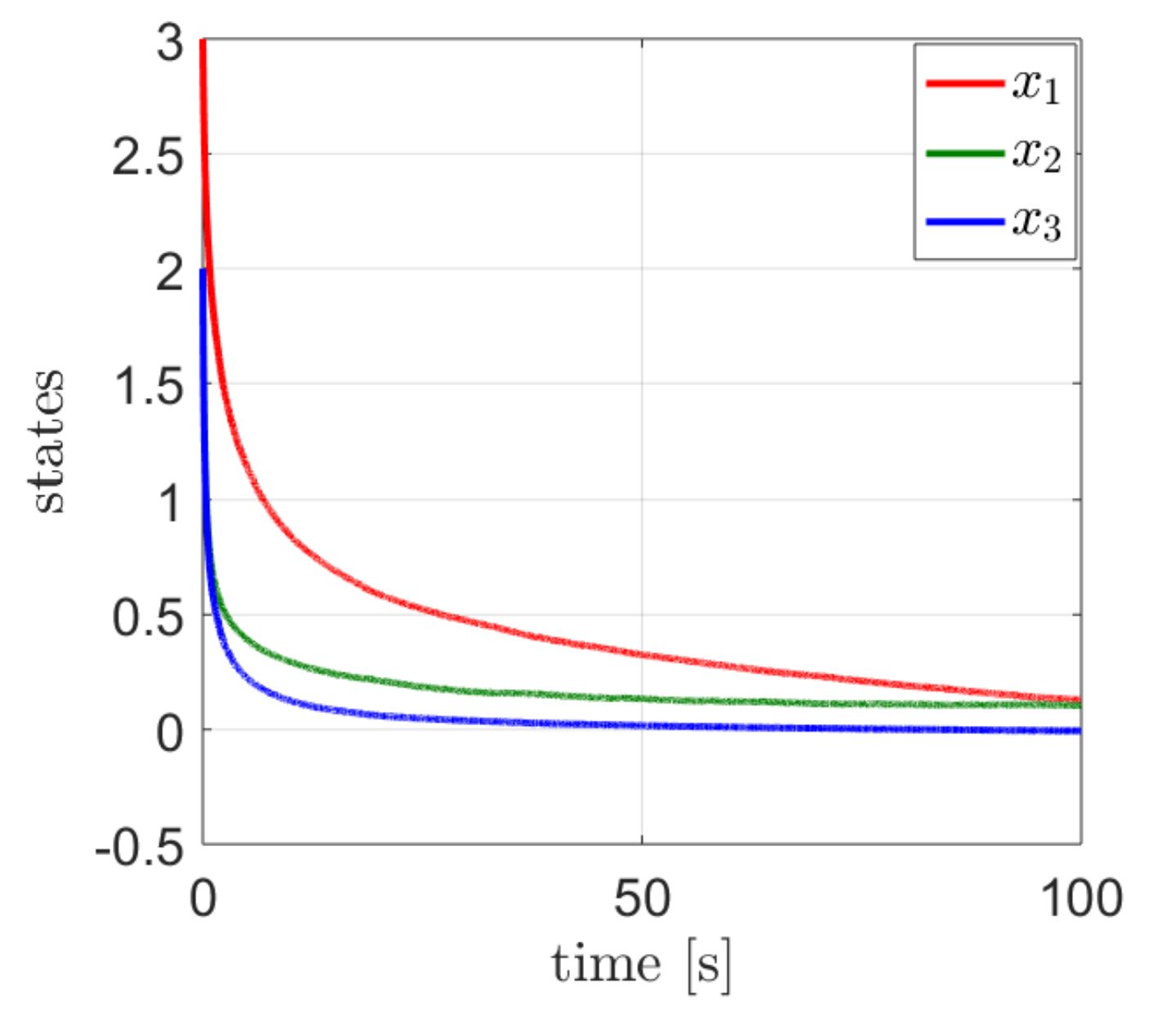}
    \label{fig:3-2}
}
\\
\subfloat[State trajectories $x(t)$ of Astolfi's controller]{
    \includegraphics[width=5cm,height=4cm]{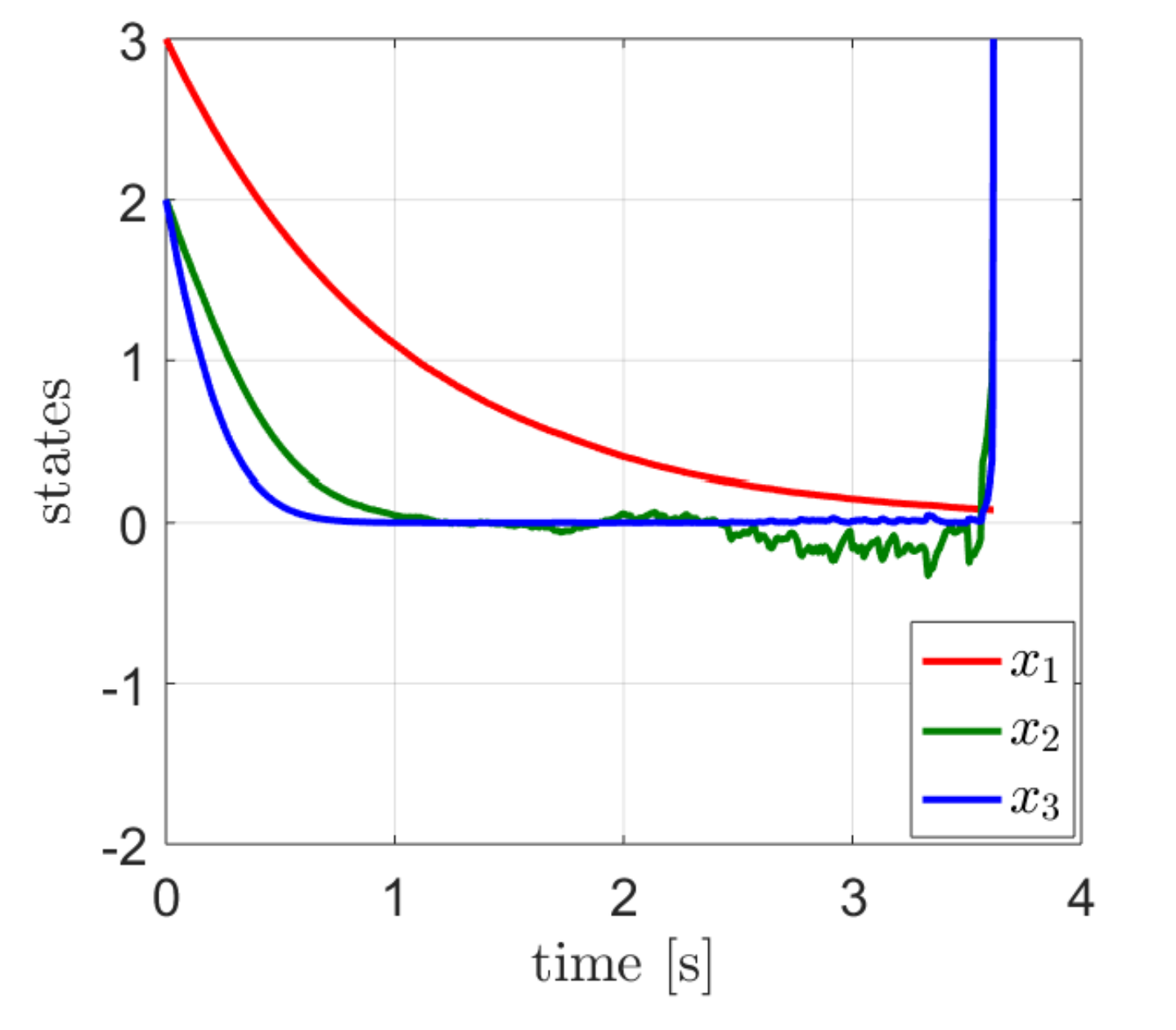}
    \label{fig:3-3}
}
\subfloat[Control inputs]{
    \includegraphics[width=5cm,height=4cm]{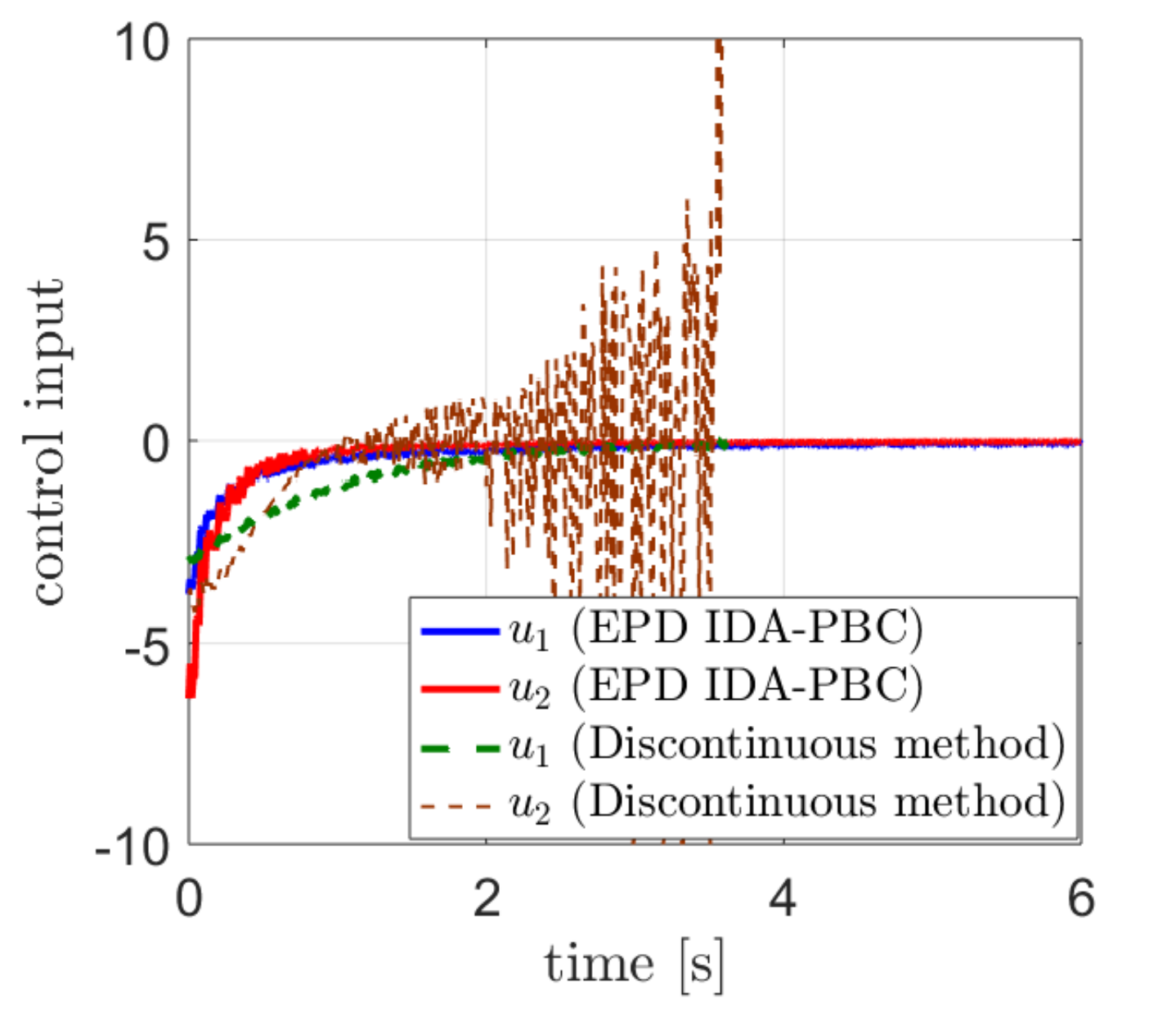}
    \label{fig:3-4}
}
\caption{Robustness evaluation of EPD IDA-PBC and Astolfi's controller for the nonholonomic integrator}
\label{fig:simulation3}
\end{figure}


\begin{figure}[h]
\centering
\subfloat[$\gamma =0.5, \beta_\ell =0.5, x(0)=(0.5, 1, 0.1, 0.5)^\top$]{
    \includegraphics[width=5cm,height=4cm]{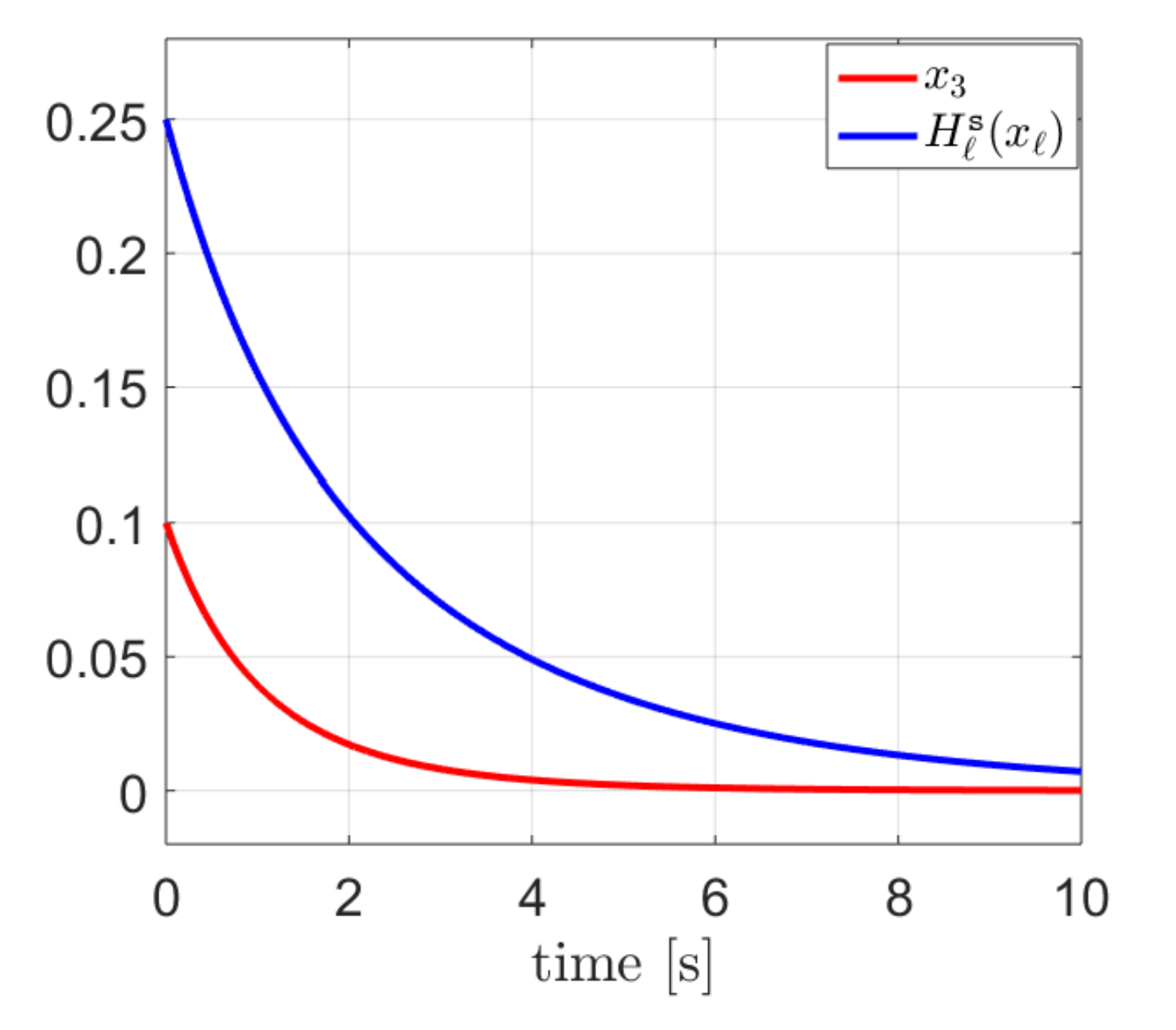}
    \label{fig:4-1}
}
\subfloat[$\gamma =0.5, \beta_\ell =0.5, x(0)=(0.5, 1, 0.1, 2)^\top$]{
    \includegraphics[width=5cm,height=4cm]{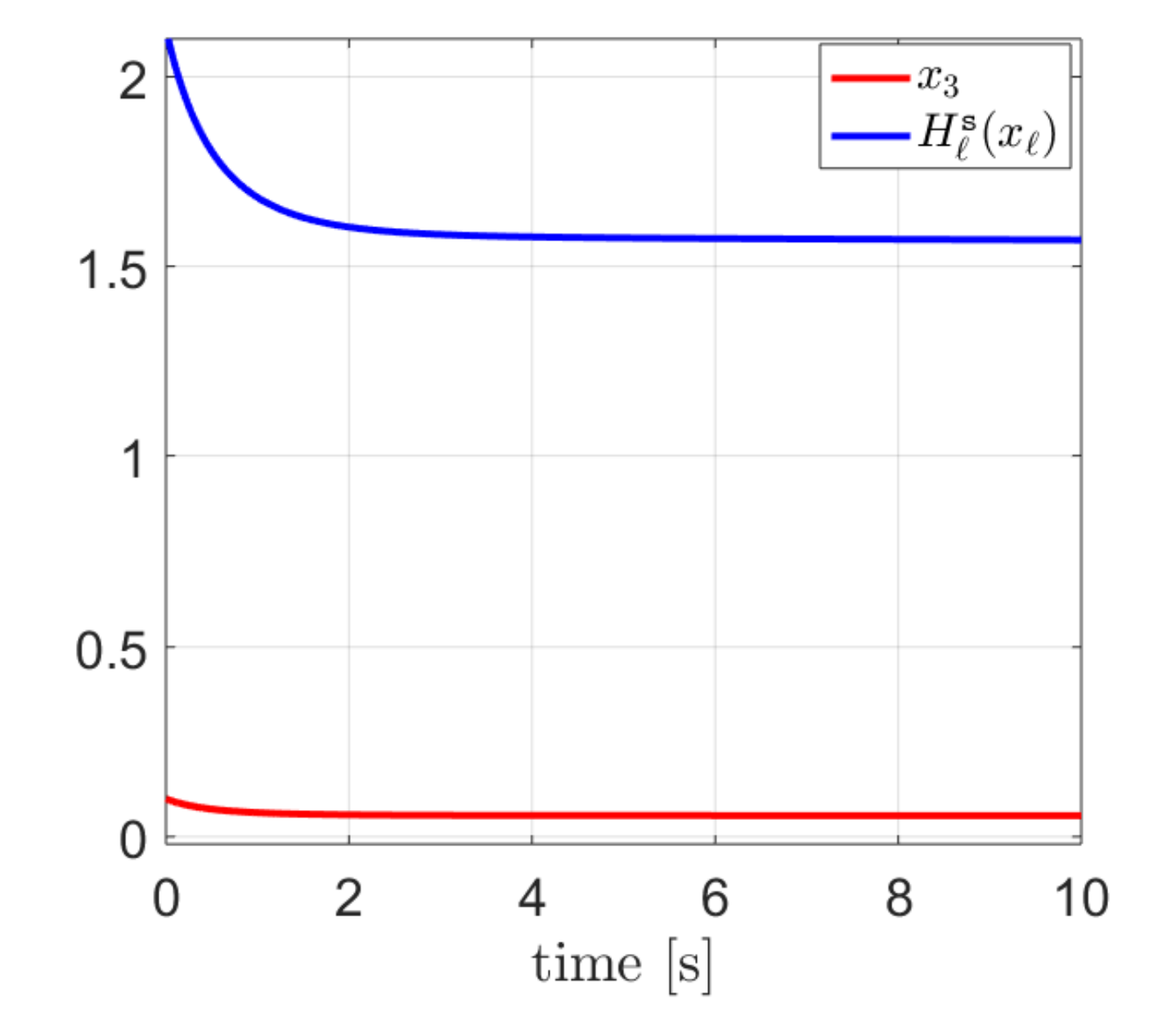}
    \label{fig:4-2}
}
\\
\subfloat[$\gamma =5, \beta_\ell =0.5, x(0)=(0.5, 1, 0.1, 1)^\top$ ]{
    \includegraphics[width=5cm,height=4cm]{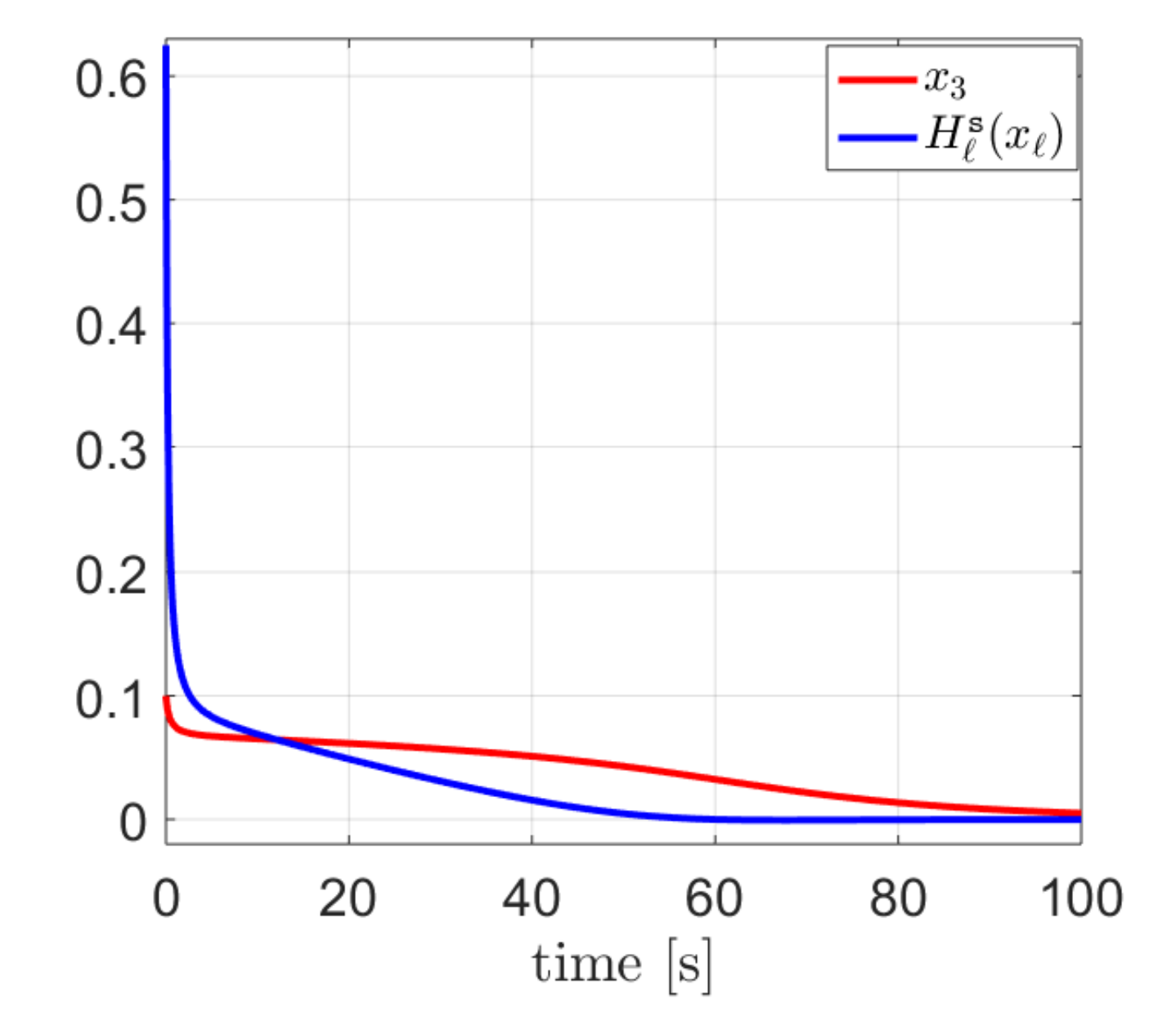}
    \label{fig:4-3}
}
\subfloat[$\gamma =50, \beta_\ell =0.5, x(0)=(0.5, 1, 0.1, 1)^\top$]{
    \includegraphics[width=5cm,height=4cm]{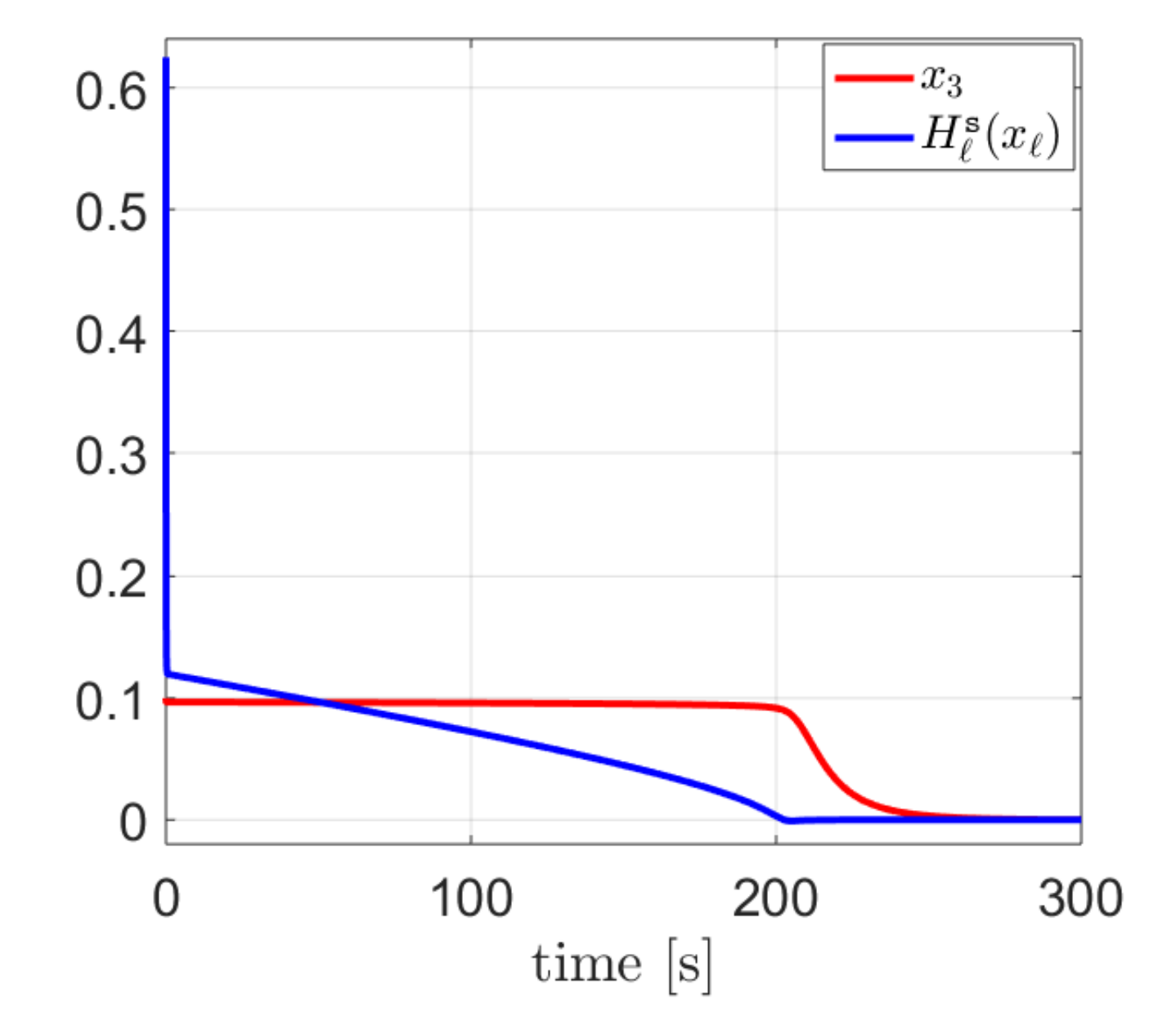}
    \label{fig:4-4}
}
\caption{EPD IDA-PBC of the nonholonomic system with chained structure and $n=4$ }
\label{fig:simulation4}
\end{figure}

\section{Concluding remarks}
\label{sec6}
We propose in this paper a variation of the well-known IDA-PBC design methodology, called EPD IDA-PBC, that is suitable for the problem of regulation of nonholonomic systems. Two asymptotic regulation objectives are considered: drive to \emph{zero} the state or drive the systems \emph{total energy} to a desired constant value. In both cases, the objectives are achieved excluding a set of inadmissible initial conditions. The main feature of this approach is that, in contrast with the existing methods reported in the literature, it yields smooth, time-invariant state-feedbacks that, in principle, have a better transient performance. This fact is illustrated via simulations.

We should also point out that in the state regulation case the zero equilibrium point is rendered stable in the sense of Lyapunov, but not asymptotically stable. We also prove in Lemma \ref{lem1} that, under the conditions of Proposition \ref{pro1}, the set of inadmissible initial conditions contains the origin. On the other hand, as indicated in Remark \ref{rem4}, for the case of energy regulation, convergence to a point ensuring the objective is achieved rendering the zero equilibrium unstable.

Current research is under way to sharpen our result for high-dimensional systems in chained structure. In particular, we are investigating alternative solutions to the matching equation \eqref{pde_general} via the proposition of different interconnection and damping matrices.

\section*{Acknowledgements}
The authors would like to express their gratitude to Prof Andrew Teel, for his careful reading of our manuscript, and many thoughtful comments that helped improve its clarity. The first author would like to thank Chi Jin (University of Waterloo) for fruitful discussions.

\appendix
\section{Proof of Lemma \ref{lem1}}
\lab{appa}
%
The proof is established by contradiction, that is, showing that if zero is not in the \emph{closure} of $\cal{I}$, then we contradict Brockett's necessary condition for asymptotic stabilization of the origin of nonholonomic systems \cite{BRO}.

In  Proposition 1 we prove that, if conditions {\bf C1-C4} hold, for all $x(0) \notin {\cal I}$, we have that
$$
\lim_{t\to\infty} x(t) =0,
$$
{\em regardless} of whether $\cal{I}$ contains the origin or not.

Denote $\calb_\delta: = \{x\in \rea^n | ~|x| \le \delta\}$. If  $\cl(\mathcal{I})$ {\em does not} contain the origin, we can always find a (small) constant $\delta_m>0$, such that
$$
\calb_{\delta_m} \cap \cal{I} = \varnothing.
$$

From the facts that $\nabla H(0) =0$ and $\nabla^2 H (x) >0$, as well as $\dot{H}\le0$, we have that the sublevel sets
$$
\cale_\epsilon: =\{x\in\rea^n | ~ H(x) \le \epsilon\}
$$
are invariant {\em for any} $\epsilon \ge 0$. Given $\delta_m>0$, we can always find small $\epsilon_m>0$ such that
$$
\cale_{\epsilon_m} \subset \calb_{\delta_m},
$$
consequently $\cale_{\epsilon_m} \cap {\cal I} = \varnothing$. Thus, we have proven the following implication
$$
x(0) \in \cale_{\epsilon_m} \quad \Rightarrow \quad
\left\{
\begin{aligned}
 & x(t)  \in \cale_{\epsilon_m}, \quad \forall t \ge 0\\
&  \lim_{t\to\infty}x(t) =0.
\end{aligned}
\right.
$$
The invariance of the set $\cale_{\epsilon_m}$ proves Lyapunov stability of the zero equilbrium, that together with the second attractivity condition,  implies {\em asymptotic stability} of the zero equilibrium of the closed-loop system. Since the controller of Proposition 1 is smooth and time invariant, this contradicts Brockett's necessary condition, completing the proof.
\qed

%

\end{document}